\documentclass[journal,final,twocolumn]{IEEEtran}
\usepackage{fixltx2e}
\usepackage{cite}
\usepackage{mathrsfs}
\usepackage{color}
\usepackage{float}

\usepackage[caption=false]{subfig}

\ifCLASSINFOpdf
   \usepackage[pdftex]{graphicx}
   \graphicspath{{Figs/}}
   \DeclareGraphicsExtensions{.pdf,.jpeg,.png}
\else
\fi

\usepackage[cmex10]{amsmath}
\usepackage{amsmath}
\usepackage{amssymb}
\usepackage{amsthm}
\usepackage{amsfonts}
\usepackage{bm}
\usepackage{xfrac}
\usepackage{empheq}
\usepackage[normalem]{ulem} 
\usepackage{soul} 
\usepackage{mathtools}

\newtheorem{theorem}{Theorem}

\theoremstyle{example}
\newtheorem{example}{Example}
\newtheorem{proposition}{Proposition}

\newtheorem{corollary}{Corollary}
\newtheorem{remark}{Remark}
\theoremstyle{definition}
\newtheorem{definition}{Definition}

\graphicspath{{figs/}}

\begin{document}
	
	\title{A Design Framework for Strongly $\chi^2$-Private Data Disclosure}
\vspace{-5mm}
\author{
		\IEEEauthorblockN{Amirreza Zamani, ~\IEEEmembership{Member,~IEEE,} Tobias J. Oechtering,~\IEEEmembership{Senior Member,~IEEE,} Mikael Skoglund,~\IEEEmembership{Fellow,~IEEE} \vspace*{0.5em}
		}
\thanks{This work was funded in
	part by the Swedish research council under contract 2019-03606. A. Zamani , M. Skoglund and T. J. Oechtering are with the division of information science and engineering, School of Electrical Engineering and
	Computer Science, KTH Royal Institute of Technology, 100 44 Stockholm,
	Sweden (e-mail: amizam@kth.se; oech@kth.se; mikael.skoglund@ee.kth.se).}}
	\maketitle

\begin{abstract}
	In this paper, we study a stochastic disclosure control problem using information-theoretic methods. The useful data to be disclosed depend on private data that should be protected. Thus, we design a 
	privacy mechanism to produce new data which maximizes the disclosed information about the useful data under a strong $\chi^2$-privacy criterion. 
	For sufficiently small leakage, the privacy mechanism design problem can be geometrically studied in the space of probability distributions by a local approximation of the mutual information. 
	 By using methods from Euclidean information geometry, the original highly challenging optimization problem can be reduced to a problem of finding the principal right-singular vector of a matrix, which characterizes the optimal privacy mechanism.
	 In two extensions we first consider a scenario where an adversary receives a noisy version of the user's message and then we look for
	 a mechanism which finds $U$ based on observing $X$, maximizing the mutual information between $U$ and $Y$ while satisfying the privacy criterion on $U$ and $Z$ under the Markov chain $(Z,Y)-X-U$. 
\end{abstract}

\section{Introduction}
The amount of data created by humans, robots, advanced cyber-physical and software systems and billions of interconnected sensors is growing rapidly. Unwanted inference possibilities from this data cause privacy threats. 
Thus, privacy mechanisms are required before data can be disclosed.

Accordingly, the information theoretic approach to privacy is receiving increased attention and related works can be found in \cite{yamamoto, sankar, makhdoumi, dwork1, oech, issa, Calmon2, asoodeh1, houi, deniz6, gun, sankar2, deniz4, asoodeh3, Calmon1, 7888175, nekouei2, Johnson,Total}. 
One of the earliest works is \cite{yamamoto}, where a source coding problem with secrecy is studied.
In both \cite{yamamoto} and \cite{sankar}, the privacy-utility trade-off is considered using expected distortion and equivocation as measure of utility and privacy. 
In \cite{makhdoumi}, the concept of a privacy funnel is introduced, where the privacy-utility trade-off under log-loss distortion is considered.
In \cite{dwork1}, the concept of differential privacy is introduced, which aims to minimize the chance of identifying the membership in a database. 
 In \cite{oech}, the hypothesis test performance of an adversary is used to measure the privacy leakage.
The concept of maximal leakage is introduced in \cite{issa} and some bounds on privacy utility trade-off are provided.  
In \cite{Calmon2}, fundamental limits of privacy utility trade-off are studied measuring the leakage using estimation-theoretic guarantees.
Properties of rate-privacy functions are studied in \cite{asoodeh1}, where either maximal correlation or mutual information are used for measuring privacy. Biometric identification systems with no privacy leakage are studied in \cite{houi}. Furthermore, recent related work, from which our work is independent\footnote{The paper \cite{jende} first appeared on arXiv significantly after our paper was submitted to these transactions.}, appears in \cite{jende}. 

Our problem formulation is closest related to \cite{deniz6} where the problem of maximizing mutual information $I(U;Y)$ given the leakage constraint $I(U;X)\leq\epsilon$ and Markov chain $X-Y-U$ is studied. 
Under the assumption of perfect privacy, i.e., $\epsilon=0$, it is shown that the privacy mechanism design problem can be reduced to a standard linear program.
In \cite{gun}, the work has been extended considering the privacy utility trade-off with a rate constraint for the disclosed data.

In this paper, we consider a similar problem as in \cite{deniz6} depicted in Fig.~\ref{fig:sysmodel}, where an agent wants to disclose some useful data to a user. The useful data is denoted by the random variable (RV) $Y$. Furthermore, $Y$ is dependent on the private data denoted by RV $X$, which is not accessible to the agent. Due to privacy considerations, the agent can not release the useful data directly. So, the agent uses a privacy mechanism to produce data $U$ that can be disclosed. $U$ should disclose as much information about $Y$ as possible and at the same time satisfy the privacy criterion. In this work, the perfect privacy condition considered in \cite{deniz6} is relaxed considering an element-wise $\chi^2$ privacy criterion which we call \textit{"Strong $\chi^2$-privacy criterion"}. A $\chi^2$-privacy criterion has been also considered in \cite{Calmon2}, studying a related privacy-utility trade-off problem. Since the optimization problem is difficult, only upper and lower bounds on the optimal privacy-utility trade-off have been derived. Furthermore,  a convex program for designing the privacy mechanism is introduced, where additional constraints are added to the main privacy problem. In contrast, we in this paper focus on finding an explicit design for the privacy mechanism problem for small leakage considering our strong $\chi^2$-privacy criterion. As a side result we show that the upper bound in \cite{Calmon2} is achievable in the small leakage regime. 

The key idea of the perfect privacy approach in \cite{deniz6} depends on revealing information aligned with the null space of the conditional distribution matrix $P_{X|Y}$. However, when the matrix is invertible this approach leads to zero utility. For instance, consider an example where a health organization center has tested multiple patients through their CD4 lymphocyte level. This test is done for finding out if an HIV positive patient is in the terminal stage of the disease (aids) or not. Assume that the results of these level tests need to be revealed for use at an aggregated level, however the organization does not want to reveal information regarding whether a person is in the terminal stage or not. In this example it can be seen that the perfect privacy approach leads to zero utility and no information can be revealed. We study this scenario in Example~3 below.


\begin{figure}[]
	\centering
	\includegraphics[width = 0.4\textwidth]{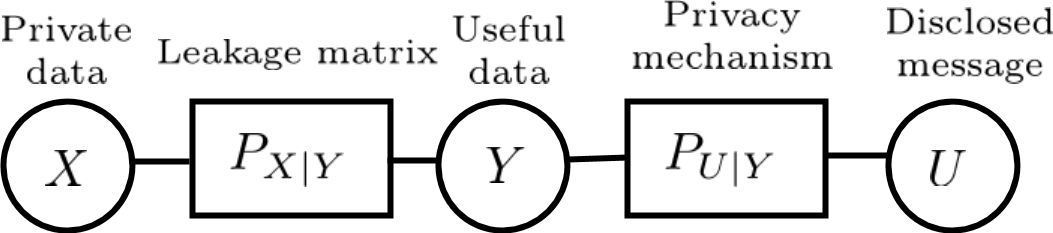}
	\caption{In this system model, disclosed data $U$ is designed by a privacy mechanism which maximizes the information disclosed about $Y$ and satisfies the strong $\chi^2$-privacy criterion.}
	\label{fig:sysmodel}
\end{figure} 

 We use methods from Euclidean information theory \cite{borade, huang} to study the design optimization problem.
 There exist many problems in information theory, where one main difficulty is not having a geometric structure on the space of probability distributions. If we assume that the distributions of interest are close to each other, then KL divergence can be well approximated by
weighted squared Euclidean distance. This results in a framework where a mutual information term has been approximated in order to simplify the optimization problem. 
This framework has been used in \cite{borade,huang}, specifically, in \cite{borade}, where it was employed for point-to-point channels and some specific broadcast channels. 
In this paper, due to the strong $\chi^2$-privacy criterion, we can exploit the information geometry approach and approximate the KL divergence and mutual information in case of a small leakage $\epsilon$. This allows us to transfer the main problem into a analytically simple largest singular value problem, which also provides deep intuitive understanding of the mechanism.

In more detail we can summerize our contribution as follows
\textbf{(i)} We present an information-theoretic disclosure control problem using a strong $\chi^2$-privacy criterion in Section  \ref{sec:system}.\\
\textbf{(ii)} We introduce and utilize concepts from Euclidean information theory to linearize the problem and derive a simple approximate solution for small leakage in Section \ref{result}. In particular our result shows that the upper bound found in \cite{Calmon2} is actually achievable for small leakage.\\
\textbf{(iii)} We provide a geometrical interpretation of the privacy mechanism design problem and two examples are given. Significantly that enhance the intuitive understanding of the privacy mechanism. In Example~1, the approximated solution and exact solution are compared. Furthermore, in Example~2 the privacy-utility trade-off with error probability as a measure of utility and estimation error as a measure of privacy is studied to compare our approach with perfect privacy. The example illustrates that by allowing an increasing privacy cost, our approach achieves higher utility, while perfect privacy results in a constant utility.    
	Finally, Example~3 illustrates a medical example where perfect privacy leads to zero utility, while our new approach allows for positive revealed information.\\
\textbf{(iv)} We transfer our methods to two extended problems which demonstrates the value of our approach as a design framework. In the first extension, a binary channel between the user and a sophisticated adversary is considered, where the agent is trying to find a mechanism to produce binary random variable $U$, which maximizes $I(U;Y)$ under the Markov chain $X-Y-U-U'$ having a privacy criterion on $X$ and $U'$. Here $U'$ is the received message by the adversary or $U$ is the stored data for future purposes and $U'$ is the first instance that is disclosed with a stored post-processing mechanism.
In the second extension, the agent looks for a mechanism which finds $U$ based on observing $X$ maximizing the mutual information between $U$ and $Y$ while satisfying the privacy criterion on $U$ and $Z$ under the Markov chain $(Z,Y)-X-U$.\\ 
 The paper is concluded in Section \ref{concul}.

\section{system model and Problem Formulation} \label{sec:system}
Let $P_{XY}$ denote the joint distribution of discrete random variables $X$ and $Y$ defined on finite alphabet $\cal{X}$ and $\cal{Y}$ with equal cardinality, i.e, $|\cal{X}|=|\cal{Y}|=\mathcal{K}$. We represent $P_{XY}$ by a matrix defined on $\mathbb{R}^{\mathcal{K}\times\mathcal{K}}$ and 
marginal distributions of $X$ and $Y$ by vectors $P_X$ and $P_Y$ defined on $\mathbb{R}^\mathcal{K}$. We assume that each element in vectors $P_X$ and $P_Y$ is non-zero. Furthermore, we represent the leakage matrix $P_{X|Y}$ by a matrix defined on $\mathbb{R}^{\cal{K}\times\cal{K}}$ which is assumed to be invertible.
In the related privacy problem with perfect privacy \cite{deniz6}, it has been shown that information can be only revealed if $P_{X|Y}$ is not invertible. This result was also proved in \cite{berger} in a source coding setup.
 RVs $X$ and $Y$ denote the private data and the useful data.
 In this work, privacy is measured by the strong $\chi^2$-privacy criterion which we introduce next.
 \begin{definition}
 	Given two random variables $X\in\mathcal{X}$ and $U\in\mathcal{U}$ with joint pmf $P_{XU}$ where $X$ describes the private data and $U$ denotes the disclosed data, for $\epsilon>0$, the \textit{strong $\chi^2$-privacy criterion} is defined as follows
 	\begin{align*}
 	&\chi^2(P_{X|U=u}||P_X)=\sum_{x\in\mathcal{X}}\frac{(P_{X|U=u}(x)-P_X(x))^2}{P_X(x)}\\&= \left\lVert[\sqrt{P_X}^{-1}](P_{X|U=u}-P_X) \right\rVert^2\leq \epsilon^2,\ \forall u\in\mathcal{U},
 	\end{align*}
 	where $[\sqrt{P_X}^{-1}]$ is a diagonal matrix with diagonal entries $\{\sqrt{P_X(x)}^{-1},x\in\mathcal{X}\}$. The norm is the Euclidean norm. The first equality is due to the definition of $\chi^2$ divergence between the two distribution vectors $P_{X|U=u}$ and $P_{X}$ \cite[Page 25]{wu2017lecture}, and the second equality is due to the definition of the $\ell_2$-norm.
 \end{definition}
The strong $\chi^2$-privacy criterion means that the all distributions (vectors) $P_{X|U=u}$ for all $u\in\mathcal{U}$ are close to $P_X$ in the Euclidean sense. 
The closeness of $P_{X|U=u}$ and $P_X$ allows us to use the concepts of information geometry so that we can locally approximate the KL divergence and mutual information between $U$ and $Y$ for small $\epsilon>0$.
In \cite{Calmon2}, the concept of $\chi^2$-information between $U$ and $X$ is employed as privacy criterion. The relation between these two criteria is as follows
\begin{align*}
\chi^2_{\text{information}}(X;U)=\mathbb{E}_U \left[\chi^2(P_{X|U=u}||P_X)\right].
\end{align*}

Our goal is to design the privacy mechanism that produces the disclosed data $U$, which maximizes $I(U;Y)$ and satisfies the strong $\chi^2$-privacy criterion. The relation between $U$ and $Y$ is described by the kernel $P_{U|Y}$ defined on $\mathbb{R}^{\mathcal{U}\times\mathcal{K}}$.
Thus, the privacy problem can be stated as follows
\begin{subequations}
\begin{align}
\sup_{P_{U|Y}} \ \ &I(U;Y),\label{privacy}\\
\text{subject to:}\ \ &X-Y-U,\label{Markov}\\
& \left\lVert[\sqrt{P_X}^{-1}](P_{X|U=u}-P_X) \right\rVert^2\leq \epsilon^2,\ \forall u\in\mathcal{U}.\label{local}
\end{align}
\end{subequations}
The strong $\chi^2$-privacy criterion for small $\epsilon$ results in closeness of $P_{Y|U=u}$ and $P_Y$ in the output distributions space, which allows us to transfer the main problem into a linear algebra problem, specifically, finding the largest singular value of a matrix. On the one hand, the criterion results in an upper bound to other privacy metrics such as KL divergence, and on the other hand, it enable us to approximately optimize the main problem.
\begin{remark}
	\normalfont
	For $\epsilon=0$ and $P_{X|Y}$  invertible, $U$ needs to be independent of $X$ and therefore $Y$ so that $I(U;Y)=0$. Accordingly, we are interested in non-trivial cases where $\epsilon$ is positive and sufficiently small \cite[Th. 4]{berger}.  
\end{remark}
\begin{remark}
	\normalfont
	By using an inequality between KL divergence and the strong $\chi^2$-privacy criterion \cite[Page 130]{wu2017lecture}, we have
	\begin{align}\label{KLchi2}
	D(P_{X|U=u}||P_X)\leq\chi^2(P_{X|U=u}||P_X)\leq\epsilon^2,\ \forall u,
	\end{align} 
	where $D(P_{X|U=u}||P_X)$ denotes KL divergence between distributions $P_{X|U=u}$ and $P_X$. Thus, we have
	\begin{align}\label{hir}
	I(U;X)=\sum_{u\in\mathcal{U}}P_U(u)D(P_{X|U=u}||P_X)\leq\epsilon^2.
	\end{align}
	Consequently, information leakage using the strong $\chi^2$-privacy criterion implies also a bound on the mutual information $I(U;X)$.  
	In the following we show that by using the Euclidean information theory method, we can strengthen \eqref{hir} and show that \eqref{local} implies $I(U;X)\leq\frac{1}{2}\epsilon^2+o(\epsilon^2)$ for small $\epsilon$.
\end{remark}
\begin{proposition}
	It suffices to consider $U$ such that $|\mathcal{U}|\leq|\mathcal{Y}|$. Furthermore, a maximum can be used in \eqref{privacy} since the
		corresponding supremum is achieved.
\end{proposition}
\begin{proof}
	The proof is provided in Appendix~A.
\end{proof}
\section{Privacy Mechanism Design}\label{result}
In this section we follow the method used in \cite{borade,huang} and show that input and output spaces can be reduced to linear spaces where the kernel describes a linear mapping between these two spaces. Thus, the privacy problem can be reduced to a linear algebra problem. The solution to the linear problem is provided and elucidates the optimal mechanism for producing $U$. We consider the resulting design framework to be the main contribution of this work.

By using \eqref{local}, we can rewrite the conditional distribution $P_{X|U=u}$ as a perturbation of $P_X$. Thus, for any $u\in\mathcal{U}$, we can write $P_{X|U=u}=P_X+\epsilon\cdot J_u$, where $J_u\in\mathbb{R}^\mathcal{K}$ is a perturbation vector that has the following three properties
\begin{align}
\sum_{x\in\mathcal{X}} J_u(x)=0,\ \forall u,\label{prop1}\\
\sum_{u\in\mathcal{U}} P_U(u)J_u(x)=0,\ \forall x\label{prop2},\\
\sum_{x\in\mathcal{X}}\frac{J_u^2(x)}{P_X(x)}\leq 1, \forall u\label{prop3}.
\end{align}
The first two properties ensure that $P_{X|U=u}$ is a valid probability distribution and the third property follows from \eqref{local}. The next proposition shows that $I(U;X)$ can be locally approximated by a squared Euclidean metric. 
In the following we use the Bachmann-Landau notation, where $o(\epsilon^2)$ describes the asymptotic behavior of a function $f:\mathbb{R}^+\rightarrow\mathbb{R}$ which satisfies that $\frac{f(\epsilon)}{\epsilon^2}\rightarrow 0$ as $\epsilon\rightarrow 0$.
\begin{proposition}\label{propp1}
	 For all $\epsilon<\frac{\min_{x\in\mathcal{X}}P_X(x)}{\sqrt{\max_{x\in\mathcal{X}}P_X(x)}}$, \eqref{local} results in a leakage constraint as follows
	\begin{align}
	I(X;U)\leq\frac{1}{2}\epsilon^2+o(\epsilon^2),\label{approx1}
	\end{align}
\end{proposition}
\begin{proof}
	The proof is provided in Appendix B.
\end{proof}
Now we show that the distribution $P_{Y|U=u}$ can be written as a linear perturbation of $P_Y$.
Since we have the Markov chain $X-Y-U$, we can write
\begin{align*}
P_{X|U=u}-P_X=P_{X|Y}[P_{Y|U=u}-P_Y]=\epsilon\cdot J_u.
\end{align*}
Due to the assumed non-singularity of the leakage matrix we obtain
\begin{align}\label{pyu}
P_{Y|U=u}-P_Y=P_{X|Y}^{-1}[P_{X|U=u}-P_X]=\epsilon\cdot P_{X|Y}^{-1}J_u.
\end{align}
The next proposition shows that $I(U;Y)$ can be locally approximated by a squared Euclidean metric \cite{borade,huang}.
	In the following we use the notation $f(x)\cong g(x)$ to describe the following equality
		\begin{align*}
		f(x)=g(x)+o(\epsilon^2),
		\end{align*}
		for a fixed $\epsilon>0$, as it will be clear from the context.
\begin{proposition}
	For all $\epsilon<\frac{|\sigma_{\text{min}}(P_{X|Y})|\min_{y\in\mathcal{Y}}P_Y(y)}{\sqrt{\max_{x\in{\mathcal{X}}}P_X(x)}}$, $I(U;Y)$ can be approximated as follows
		\begin{align}
	I(Y;U)\cong\frac{1}{2}\epsilon^2\sum_u P_U\|[\sqrt{P_Y}^{-1}]P_{X|Y}^{-1}[\sqrt {P_X}]L_u\|^2,\label{approx2}
	\end{align}
	where $[\sqrt{P_Y}^{-1}]$ and $[\sqrt{P_X}]$ are diagonal matrices with diagonal entries $\{\sqrt{P_Y}^{-1},\ \forall y\in\mathcal{Y}\}$ and $\{\sqrt{P_X},\ \forall x\in\mathcal{X}\}$. Furthermore, for every $u\in\mathcal{U}$ we have $L_u=[\sqrt{P_X}^{-1}]J_u\in\mathbb{R}^{\mathcal{K}}$.
\end{proposition}
\begin{proof}
	For the local approximation of the KL-divergence we follow similar arguments as in \cite{borade,huang}:
	\begin{align*}
	I(Y;U)&=\sum_u P_U(u)D(P_{Y|U=u}||P_Y)\\&=\sum_u P_U(u)\sum_y P_{Y|U=u}(y)\log\left(\frac{P_{Y|U=u}(y)}{P_Y(y)}\right)\\&=\sum_u P_U(u)\sum_y\! P_{Y|U=u}(y)\log\left(\!1\!+\!\epsilon\frac{P_{X|Y}^{-1}J_u(y)}{P_Y(y)}\right)\\&\stackrel{(a)}{=}\frac{1}{2}\epsilon^2\sum_u P_U\sum_y
	\frac{(P_{X|Y}^{-1}J_u)^2}{P_Y}+o(\epsilon^2)\\
	&=\frac{1}{2}\epsilon^2\sum_u P_U\|[\sqrt{P_Y}^{-1}]P_{X|Y}^{-1}J_u\|^2+o(\epsilon^2)
	\\&\cong\frac{1}{2}\epsilon^2\sum_u P_U\|[\sqrt{P_Y}^{-1}]P_{X|Y}^{-1}[\sqrt{P_X}]L_u\|^2,
\end{align*}
where (a) comes from second order Taylor expansion of $\log(1+x)$ which is equal to $x-\frac{x^2}{2}+o(x^2)$ and using the fact that we have $\sum_y P_{X|Y}^{-1}J_u(y)=0$. The latter follows from \eqref{prop1} and the property of the leakage matrix $\bm{1}^T\cdot P_{X|Y}=\bm{1}^T$, we have
\begin{align*}
	0=\bm{1}^T\cdot J_u=\bm{1}^T\cdot P_{X|Y}^{-1}J_u,  
\end{align*}
where $\bm{1}\in\mathbb{R}^{\mathcal{K}}$ denotes a vector with all entries equal to $1$. For approximating $I(U;Y)$, we use the second Taylor expansion of $\log(1+x)$. Therefore we must have $|\epsilon\frac{P_{X|Y}^{-1}J_u(y)}{P_Y(y)}|<1$ for all $u$ and $y$. One sufficient condition for $\epsilon$ to satisfy this inequality is to have $\epsilon<\frac{|\sigma_{\text{min}}(P_{X|Y})|\min_{y\in\mathcal{Y}}P_Y(y)}{\sqrt{\max_{x\in{\mathcal{X}}}P_X(x)}}$, since in this case we have
\begin{align*}
\epsilon^2|P_{X|Y}^{-1}J_u(y)|^2&\leq\epsilon^2\left\lVert P_{X|Y}^{-1}J_u\right\rVert^2\leq\epsilon^2 \sigma_{\max}^2\left(P_{X|Y}^{-1}\right)\left\lVert J_u\right\rVert^2\\&\stackrel{(a)}\leq\frac{\epsilon^2\max_{x\in{\mathcal{X}}}P_X(x)}{\sigma^2_{\text{min}}(P_{X|Y})}<\min_{y\in\mathcal{Y}} P_Y^2(y),
\end{align*}
which implies $|\epsilon\frac{P_{X|Y}^{-1}J_u(y)}{P_Y(y)}\!|<1$.
The step (a) follows from $\sigma_{\max}^2\left(P_{X|Y}^{-1}\right)=\frac{1}{\sigma_{\min}^2\left(P_{X|Y}\right)}$ and $\|J_u\|^2\leq\max_{x\in{\mathcal{X}}}P_X(x)$. The latter inequality follows from \eqref{prop3} since we have
\begin{align*}
\frac{\|J_u\|^2}{\max_{x\in{\mathcal{X}}}P_X(x)}\leq \sum_{x\in\mathcal{X}}\frac{J_u^2(x)}{P_X(x)}\leq 1.
\end{align*}
\end{proof}
	\begin{remark}
		In Proposition~2 and Proposition~3, approximation of mutual information is based on a Taylor expansion of $\log(1+x)$, where the expansion converges for all $|x|<1$. The upper bounds on $\epsilon$ are derived from this constraint, i.e., for all those $\epsilon$ the Taylor expansion of KL-divergence converges.
\end{remark}
The following result shows that by using local approximation in \eqref{approx2}, the privacy problem defined in \eqref{privacy} can be reduced to a linear algebra problem. In more detail, by substituting $L_u$ in \eqref{prop1}, \eqref{prop2} and \eqref{prop3} we obtain next corollary. 
\begin{corollary}\label{corr1}
	For all $\epsilon<\frac{|\sigma_{\text{min}}(P_{X|Y})|\min_{y\in\mathcal{Y}}P_Y(y)}{\sqrt{\max_{x\in{\mathcal{X}}}P_X(x)}}$, the privacy mechanism design problem in \eqref{privacy} can be approximately solved by the following linear problem
	\begin{align}
	\max_{\{L_u,P_U\}} \ &\sum_u P_U(u)\|W\cdot L_u\|^2,\label{newprob2}\\
	\text{subject to:}\ &\|L_u\|^2\leq 1,\ \forall u\in\mathcal{U},\\
	&\sum_x \sqrt{P_X(x)}L_u(x)=0,\ \forall u,\label{orth}\\
	&\sum_u P_U(u)\sqrt{P_X(x)}L_u(x)=0,\ \forall x,\label{orth2}
	\end{align}
	where $W=[\sqrt{P_Y}^{-1}]P_{X|Y}^{-1}[\sqrt{P_X}]$ and the $o(\epsilon^2)$-term is ignored.
\end{corollary}
		Condition \eqref{orth} can be interpreted as an inner product between vectors $L_u$ and $\sqrt{P_X}$, where$\sqrt{P_X}\in\mathbb{R}^{\mathcal{K}}$ is a vector with entries $\{\sqrt{P_X(x)},\ x\in\mathcal{X}\}$. Thus, condition \eqref{orth} states an orthogonality condition. 
		Furthermore, \eqref{orth2} can be rewritten in vector form as $\sum_u P_U(u)L_u=\bm 0\in\mathcal{R}^{\mathcal{K}}$ using the assumption that $P_X(x)>0$ for all $x\in\mathcal{X}$.
	Therewith, the problem in Corollary \ref{corr1} can be rewritten as
	\begin{align}
	\max_{\begin{array}{c}   \substack{L_u,P_U:\|L_u\|^2\leq 1,\\ L_u\perp\sqrt{P_X},\\ \sum_u P_U(u)L_u=\bm 0} \end{array}} \sum_u P_U(u)\|W\cdot L_u\|^2.\label{max}
	\end{align}
	The next proposition shows how to simplify \eqref{max}.
	\begin{proposition}
		Let $L^*$ be the maximizer of \eqref{max2}, then \eqref{max} and \eqref{max2} achieve the same maximum value while $U$ as a uniform binary RV with $L_0=-L_1=L^*$ maximizes \eqref{max}.
		\begin{align}
		\max_{L:L\perp \sqrt{P_X},\ \|L\|^2\leq 1} \|W\cdot L\|^2.\label{max2}
		\end{align}
		\begin{proof}
			Let $\{L_u^*,P_U^*\}$ be the maximizer of \eqref{max}. Furthermore, let $u'$ be the index that maximizes $\|W\cdot L_{u}^*\|^2$, i.e., $u'=\text{argmax}_{u\in\mathcal{U}} \|W\cdot L_{u}^*\|^2$. Then we have

			\begin{align*}
			\sum_u P_U^*(u)||W\cdot L_u^*||^2\leq ||W\cdot L_{u'}^*||^2\leq||W\cdot L^*||^2,
			\end{align*}
			where the right inequality comes from the fact that $L^*$ has to satisfy one less constraint than $L_{u'}^*$. However, by choosing $U$ as a uniform binary RV and $L_0=-L_1=L^*$ the constraints in \eqref{max} are satisfied and the maximum in \eqref{max2} is achieved. Thus, without loss of optimality we can choose $U$ as a uniformly distributed binary RV and \eqref{max} reduces to \eqref{max2}.
		\end{proof}
	\end{proposition} 
	After finding the solution of \eqref{max2}, the conditional distributions $P_{X|U=u}$ and $P_{Y|U=u}$ are given by
	\begin{align}
	P_{X|U=0}&=P_X+\epsilon[\sqrt{P_X}]L^*,\label{cond11}\\
	P_{X|U=1}&=P_X-\epsilon[\sqrt{P_X}]L^*,\label{cond12}\\
	P_{Y|U=0}&=P_Y+\epsilon P_{X|Y}^{-1}[\sqrt{P_X}]L^*,\label{condis1}\\
	P_{Y|U=1}&=P_Y-\epsilon P_{X|Y}^{-1}[\sqrt{P_X}]L^*.\label{condis2}
	\end{align}

In next theorem we derive the solution of \eqref{max2}.
\begin{theorem}\label{th1}
	$L^*$, which maximizes \eqref{max2}, is the right singular vector corresponding to the largest singular value of $W$.
\end{theorem}
\begin{proof}
	The proof is provided in Appendix C.
\end{proof}
By using Theorem~\ref{th1}, the solution to the problem in Corollary~\ref{corr1} can be summarized as $\{P_U^*,L_u^*\}=\{U\ \text{uniform binary RV},\ L_0=-L_1=L^*\}$, where $L^*$ is the solution of \eqref{max2}. Thus, we have the following result.
\begin{corollary}
	The maximum value in \eqref{privacy} can be approximated by $\frac{1}{2}\epsilon^2\sigma_{\text{max}}^2$ for small $\epsilon$ and can be achieved by a privacy mechanism characterized by the conditional distributions found in \eqref{condis1} and \eqref{condis2},
	 where $\sigma_{\text{max}}$ is the largest singular value of $W$ corresponding to the right singular vector $L^*$.
	\end{corollary}
	\begin{figure*}[]
		\centering
		\includegraphics[width = .7\textwidth]{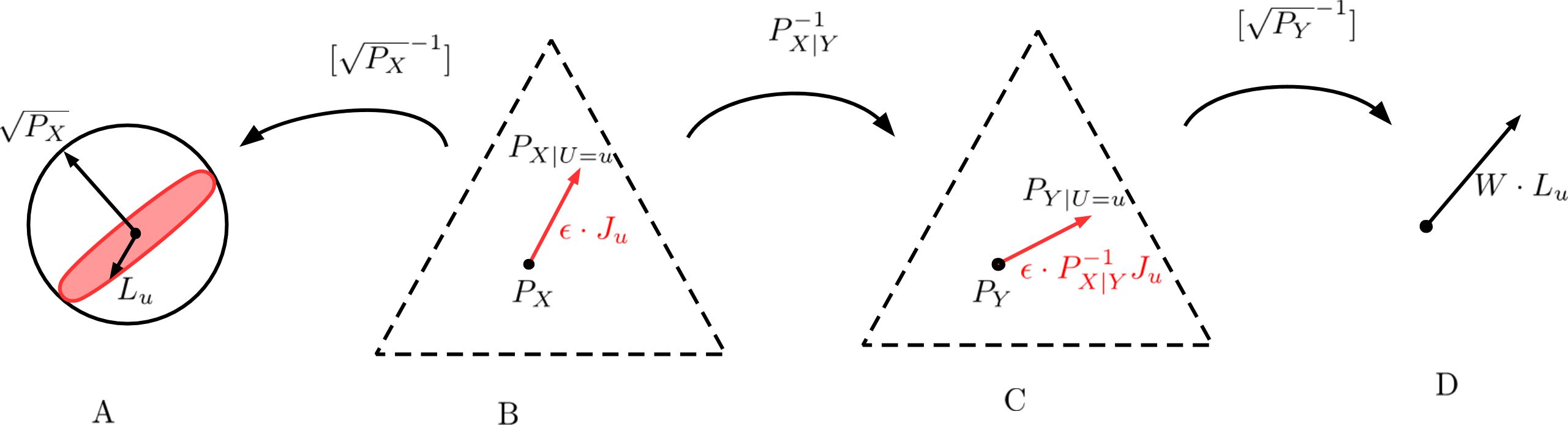}
		\caption{ For the privacy mechanism design, we are looking for $L^*$ in the red region (vector space A) which results in a vector with the largest Euclidean norm in vector space D. Space B and space C are probability spaces for the input and output distributions, the circle in space A represents the vectors that satisfy the strong $\chi^2$-privacy criterion and the red region denotes all vectors that are orthogonal to vector $\sqrt{P_X}$. Starting from Space A and reaching Space D the mapping between Space A and Space D can be found as  $W=[\sqrt{P_Y}^{-1}]P_{X|Y}^{-1}[\sqrt{P_X}]$.}
		\label{fig:inter}
	\end{figure*}
		\section{ geometric interpretation and discussions}\label{dissc}
	In Figure~\ref{fig:inter}, four spaces are illustrated. Space B and space C are probability spaces of the input and output distributions, where the points are inside a simplex. Multiplying input distributions by $P_{X|Y}^{-1}$ results in output distributions. Space A illustrates vectors $L_u$ with norm smaller than 1, which corresponds to the strong $\chi^2$-privacy criterion. The red region in this space includes all vectors that are orthogonal to $\sqrt{P_X}$. 
	For the optimal solution with $U$ chosen to be a equiprobable binary RV, it is shown that it remains to find the vector $L_u$ in the red region that results in a vector that has the largest norm in space D. This is achieved by the principal right-singular vector of $W$.
	The mapping between space A and B is given by $[\sqrt{P_X}^{-1}]$ and also the mapping between space C and D is given by $[\sqrt{P_Y}^{-1}]$. Thus $W$ is given by $[\sqrt{P_Y}^{-1}]P_{X|Y}^{-1}[\sqrt{P_X}]$. 
	
	In following, we provide an example where the procedure of finding the mechanism to produce $U$ is illustrated.
\begin{example}
	\normalfont
	Consider the leakage matrix $P_{X|Y}=\begin{bmatrix}
	\frac{1}{4}       & \frac{2}{5}  \\
	\frac{3}{4}    & \frac{3}{5}
	\end{bmatrix}$
	 and $P_Y$ is given as $[\frac{1}{4} , \frac{3}{4}]^T$. Thus, we can calculate $W$ and $P_X$ as
	 \begin{align*}
	 P_X&=P_{X|Y}P_Y=[0.3625, 0.6375]^T,\\
	 W &= [\sqrt{P_Y}^{-1}]P_{X|Y}^{-1}[\sqrt{P_X}] = \begin{bmatrix}
	 -4.8166       & 4.2583  \\
	 3.4761    & -1.5366
	 \end{bmatrix}.
	 \end{align*}
	 The singular values of $W$ are $7.4012$ and $1$ with corresponding right singular vectors $[0.7984, -0.6021]^T$ and $[0.6021 , 0.7954]^T$, respectively. Thus the maximum of \eqref{privacy} is approximately $\frac{1}{2}\epsilon^2(7.4012)^2=27.39\cdot \epsilon^2$.
	  
	 The maximizing vector $L^*$ in \eqref{max2} is equal to $[0.7984, -0.6021]^T$ and the mapping between $U$ and $Y$ can be calculated as follows (the approximate maximum of $I(U;Y)$ is achieved by the following conditional distributions):
	 \begin{align*}
	 P_{Y|U=0}&=P_Y+\epsilon P_{X|Y}^{-1}[\sqrt{P_X}]L^*\\
	 &=[0.25-3.2048\cdot\epsilon , 0.75+3.2048\cdot\epsilon]^T,\\
	 P_{Y|U=1}&=P_Y-\epsilon P_{X|Y}^{-1}[\sqrt{P_X}]L^*\\
	 &=[0.25+3.2048\cdot\epsilon , 0.75-3.2048\cdot\epsilon]^T.
	 \end{align*}
	 Note that the approximation is valid if $|\epsilon \frac{P_{X|Y}^{-1}J_{u}(y)}{P_{Y}(y)}|\ll 1$ holds for all $y$ and $u$ .
	 For the example above we have $\epsilon\cdot P_{X|Y}^{-1}J_0=\epsilon[-3.2048,\ 3.2048]^T$ and $\epsilon\cdot P_{X|Y}^{-1}J_1=\epsilon[3.2048,\ -3.2048]^T$ so that $\epsilon\ll 0.078$. Fig.~\ref{fig:shsh} illustrates the exact solution of \eqref{privacy} and the proposed approximated solution, i.e., $\frac{1}{2}\epsilon^2\sigma_{\max}^2$, where the exact solution is found by exhaustive search. As can be seen the approximation error of our proposed solution in the high privacy regime, i.e., small $\epsilon$, is negligible.
	 \begin{figure*}[]
	 	\centering
	 	\includegraphics[width = .8\textwidth]{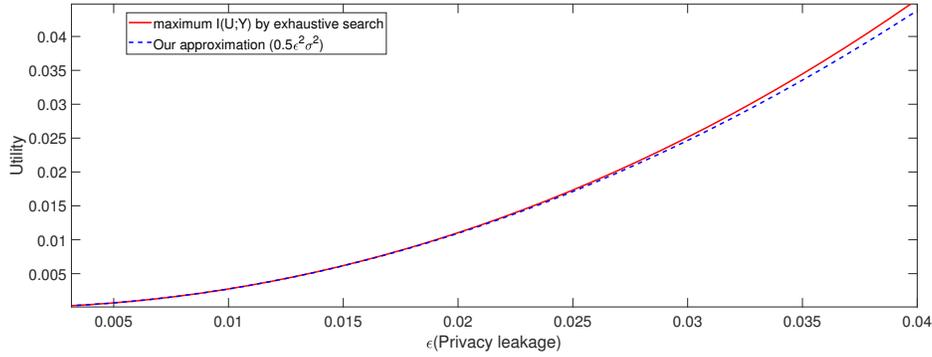}
	 	\caption{Comparing our proposed approximate solution in Example~1 with the exact solution found by an exhaustive search. It can be seen that when small leakage is allowed the approximated utility is very close to the exact utility given by $I(U;Y)$.}
	 	\label{fig:shsh}
	 \end{figure*}
	\end{example}
	In next example we consider a BSC($\alpha$) channel as leakage matrix. We provide an example with a constant upper bound on the approximated mutual information. Furthermore, the privacy-utility trade-off is studied in two scenarios, where in the first scenario we have used our approach to design $P_{U|Y}$ and in the second scenario the approach in \cite{deniz6} is used. 
		Here, utility is measured by probability of error between disclosed and desired data and estimation error measures the privacy. We show how our approach achieves better utility when small leakage is allowed, however the perfect privacy approach in \cite{deniz6} attains constantly an error probability of one-half, i.e., no utility.
	\begin{example}
		\normalfont
		Let $P_{X|Y}=\begin{bmatrix}
		1-\alpha       & \alpha  \\
		\alpha    & 1-\alpha
		\end{bmatrix}$ and $P_Y$ is given by $[\frac{1}{4} , \frac{3}{4}]^T$. By following the same procedure we have
		\begin{align*}
			P_X&=P_{X|Y}P_Y=[\frac{2\alpha+1}{4}, \frac{3-2\alpha}{4}]^T,\\
			W &= [\sqrt{P_Y}^{-1}]P_{X|Y}^{-1}[\sqrt{P_X}] \\&= \begin{bmatrix}
				\frac{\sqrt{2\alpha+1}(\alpha-1)}{(2\alpha-1)}       & \frac{\alpha\sqrt{3-2\alpha}}{2\alpha-1}  \\
				\frac{\alpha\sqrt{2\alpha+1}}{\sqrt{3}(2\alpha-1)}    & \frac{\sqrt{3-2\alpha}(\alpha-1)}{\sqrt{3}(2\alpha-1)}
			\end{bmatrix}.
		\end{align*}
		Singular values of $W$ are $\sqrt{\frac{(2\alpha+1)(3-2\alpha)}{3(2\alpha-1)^2}}\geq 1$ for $\alpha\in[0,\frac{1}{2})$ and $1$ with corresponding right singular vectors $[-\sqrt{\frac{3-2\alpha}{4}},\ \sqrt{\frac{2\alpha+1}{4}}]^T$ and $[\sqrt{\frac{2\alpha+1}{4}}, \sqrt{\frac{3-2\alpha}{4}}]^T$, respectively. Thus, we have $L^*=[-\sqrt{\frac{3-2\alpha}{4}},\ \sqrt{\frac{2\alpha+1}{4}}]^T$ and $\max I(U;Y)\approx \epsilon^2\frac{(2\alpha+1)(3-2\alpha)}{6(2\alpha-1)^2}$ with the following conditional distributions
		\begin{align}
			P_{Y|U=0}&=P_Y+\epsilon\cdot P_{X|Y}^{-1}[\sqrt{P_X}]L^*\nonumber\\&=[\frac{1}{4}\!\!+\!\epsilon\frac{\sqrt{(3\!-\!2\alpha)(2\alpha\!+\!1)}}{4(2\alpha\!-\!1)},\ \!\!\!\frac{3}{4}\!\!-\!\epsilon\frac{\sqrt{(3\!-\!2\alpha)(2\alpha\!+\!1)}}{4(2\alpha\!-\!1)}],\label{sossher}\\
			P_{Y|U=1}&=P_Y-\epsilon\cdot P_{X|Y}^{-1}[\sqrt{P_X}]L^*\nonumber\\&=[\frac{1}{4}\!\!-\!\epsilon\frac{\sqrt{(3\!-\!2\alpha)(2\alpha\!+\!1)}}{4(2\alpha\!-\!1)},\ \!\!\!\frac{3}{4}\!\!+\!\epsilon\frac{\sqrt{(3\!-\!2\alpha)(2\alpha\!+\!1)}}{4(2\alpha\!-\!1)}].
		\end{align}
		The approximation of $I(U;Y)$ holds when we have $|\epsilon\frac{P_{X|Y}^{-1}[\sqrt{P_X}]L^*}{P_Y}|\ll 1$ for all $y$ and $u$, which leads to $\epsilon\ll\frac{|2\alpha-1|}{\sqrt{(3-2\alpha)(2\alpha+1)}}$. If $\epsilon<\frac{|2\alpha-1|}{\sqrt{(3-2\alpha)(2\alpha+1)}}$, then the approximation of the  mutual information $I(U;Y)\cong\frac{1}{2}\epsilon^2\sigma_{\text{max}}^2$ is upper bounded by $\frac{1}{6}$ for all $0\leq\alpha<\frac{1}{2}$.

			Next, we consider two scenarios, where in first scenario our approach is used to find the sub-optimal Kernel $P_{U|Y}$ and in the second scenario the perfect privacy approach in \cite{deniz6} is used. Intuitively, $U$ and $Y$ should be as much correlated as possible under the privacy constraint. Furthermore, mutual information increases with lower probability of error. Thus, we consider the probability of error between disclosed data $U$ and desired data $Y$ as utility, i.e., $P_r(Y\neq U)$, and the MMSE of estimating $X$ based on observing $U$ as privacy measure in this example.
			\begin{figure*}[h]
				\centering
				\includegraphics[width = .8\textwidth]{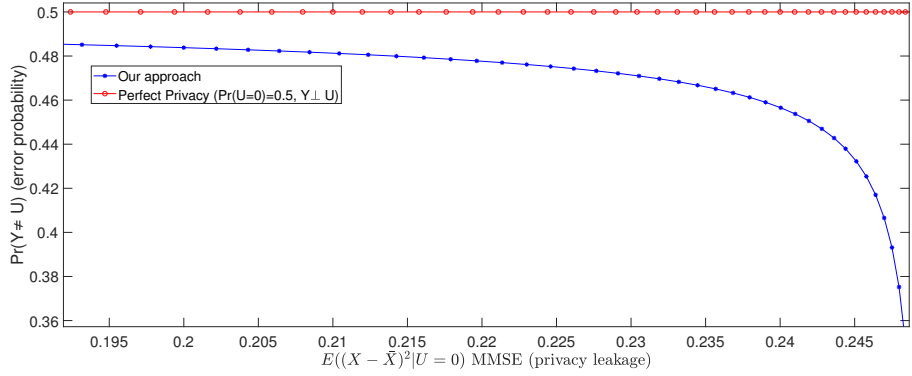}
				\caption{Privacy-utility trade-off for designed $U$ based on our method and the perfect privacy approach in \cite{deniz6} resulting in an independent $U$. It can be seen that when higher leakage is allowed our approach achieves increased utility. In the first scenario using our approach, $\bar{X}=\hat{X}$ and in the second scenario, $\bar{X}=\tilde{X}$ assuming $U=0$ is observed.}
				\label{fig:kosmos1}
			\end{figure*}
			In Fig.~\ref{fig:kosmos1}, the parameter $\alpha\neq\frac{1}{2}$, which corresponds to correlation between $X$ and $Y$, is swept to illustrate the privacy-utility trade-off. Furthermore, $\hat{X} = \mathbb{E}\{X|U=u\}$ is the estimation of $X$ based on observing $U$ using \eqref{cond11} with uniform $U$ and $\tilde{X} =  \mathbb{E}\{X|U=u\}=\mathbb{E}\{X\},$ is the estimation of $X$ based on observing $U$ for the perfect privacy approach, where $U$ has to be independent of $X$ and $Y$ in order to satisfy the perfect privacy constraint, i.e., for $\epsilon=0$ and invertible leakage matrix we have $I(U;Y)=0$.
			Thus, the estimation error can be calculated as follows
			\begin{align*}
			\mathrm{MMSE}(\hat{X})&=\mathbb{E}\{(X-\hat{X})^2|U=0\}\\&=\mathbb{E}\{X^2|U=0\}-\left(\mathbb{E}\{X|U=0\}\right)^2\\&=P(X=1|U=0)-P(X=1|U=0)^2,\\
			\mathrm{MMSE}(\tilde{X})&=\mathbb{E}\{(X-\tilde{X})^2|U=0\}\\&=\mathbb{E}\{X^2\}-\left(\mathbb{E}\{X\}\right)^2\\&=P(X=1)-P(X=1)^2,
			\end{align*}
			where the conditional distribution $P(X=1|U=0)$ is found by \eqref{cond11} and in last line we used the independency of $X$ and $U$.  
			Furthermore, the error probability in both cases can be calculated as
			\begin{align*}
			&P_{\chi^2}(U\neq Y)=P(U=0,Y=1)+P(U=1,Y=0)\\&=\frac{1}{2}\left(P(Y=0|U=1)+P(Y=1|U=0)\right),\\
			&P_{\epsilon=0}(U\neq Y)=P(Y=0)P(U=1)\!+\!P(Y=1)P(U=0),
			\end{align*}   
			where $P_{\chi^2}(U\neq Y)$ is the error probability of our approach and $P_{\epsilon=0}(U\neq Y)$ is the error probability of the perfect privacy approach. 
			The conditional distribution $P(Y=0|U=1)$ is found by \eqref{sossher} and for the perfect privacy scenario, we assumed $U$ is uniformly distributed.	\\
			The privacy-utility trade-off employing error probability and MMSE is illustrated in Fig.~\ref{fig:kosmos1}. As it can be seen when the leakage is increasing, the error probability of our approach decreases which means we achieve better utility, i.e., larger mutual information and lower error probability, however the approach in \cite{deniz6} attains a constant utility (constant probability of error equals to one-half which corresponds to zero mutual information). In other words, allowing a small privacy leakage can have a significant utility gain achieved with our design approach.
	\end{example}
	The next example outlines scenarios where a privacy filter will be essential to enable data sharing. This applies to Example~2 as well.
	\begin{example}
		\normalfont
		Consider a situation with a binary private variable $X\in\{0,1\}$. Assume the variable is hidden but we can observe $Z$ that depends it, and let $Y\in\{0,1\}$ be the outcome of a statistical test based on observing $Z$. For example, $X$ can be ``driver is drunk," $Z$ correlates with ``driving behavior" and $Y$ is ``driver is exceeding the speed limit."	Or consider instead a scenario formed by patients under treatment. Here a specific example would be that $Z$ is the so called CD4 lymphocyte level measured for a specific patient. When the value is low, then with high probability the patient is in the terminal stage of an HIV infection. However this is not the single cause for low CD4 values (for example being under chemotherapy would be another cause). On the other hand, a patient in the terminal stage of HIV can also have higher CD4 values. Then let $Y$ corresponds to ``CD4 is low/high" and $X$ to ``patient is in terminal stage of HIV." To get a numerical example, assume that ``low" corresponds to $Z<200$ and $Y=1$, and that Pr$(X=1|Y=1)=0.9.$ Also assume that the likelihood of the patient being in the terminal stage of HIV given $Y=0$ is $0.05$, i.e., Pr$(X=1|Y=0)=0.05.$ The corresponding leakage matrix $P_{X|Y}$ is clearly invertible. Assume that in a study 6 percent of the patients are observed to have low CD4 values, i.e., Pr$(Y=1)=0.06$, which results Pr$(X=1)=0.101.$ Then the optimal $W$ matrix can be computed as  
		\begin{align*}
		W=\begin{bmatrix}
		1.4501  & -0.2277\\
		-0.0386  &  1.0355
		\end{bmatrix}.
		\end{align*}
		Since the leakage matrix is invertible, perfect privacy would result in zero utility. However, by using our approach we obtain $1.1141\epsilon^2$ utility or $1.6073\epsilon^2$ bits utility.   
\end{example}
	Our next result discusses the relation to \cite{Calmon2}. While the focus in the present paper is on introducing the proposed design framework, our result also shows that an upper bound in \cite{Calmon2} is actually achievable since we can achieve it considering even a strong privacy criterion. In order to compare the results one needs to substitute $S$, $X$, $Y$ and $\epsilon$ in \cite{Calmon2}, by $X$, $Y$, $U$ and $\epsilon^2$, respectively.
	\begin{proposition}
		For all $\epsilon\!<\!\frac{|\sigma_{\text{min}}(P_{X|Y})|\min_{y\in\mathcal{Y}}P_Y(y)}{\sqrt{\max_{x\in{\mathcal{X}}}P_X(x)}}$, the upper bound on the privacy-utility trade-off derived in \cite[Th.2]{Calmon2}, is tight. 
	\end{proposition}
\begin{proof}
	First we show that the approximation of $I(U;Y)$ found in \eqref{approx2}, is equal to half of the $\chi^2$-information between $U$ and $Y$. By using Proposition~2, we have
	\begin{align*}
	&\frac{1}{2}\epsilon^2\sum_u P_U\|[\sqrt{P_Y}^{-1}]P_{X|Y}^{-1}[\sqrt{P_X}]L_u\|^2=\\
	&\frac{1}{2}\epsilon^2\sum_u \!P_U\!\sum_y\frac{(P_{X|Y}^{-1}J_u)^2}{P_Y}\!=\!\frac{1}{2}\!\sum_u \!P_U\!\!\sum_y\frac{(P_{Y|U=u}-P_Y)^2}{P_Y}=\\
	&\frac{1}{2}\sum_u P_U \chi^2(P_{Y|U=u}||P_Y)=\frac{1}{2}\chi^2_{\text{information}}(Y;U).
	\end{align*}
	Thus, the problem found in \eqref{newprob2}, is equivalent to the following problem
	\begin{subequations}
	\begin{align}
	\max_{P_{U|Y}} \ \ &\chi^2_{\text{information}}(Y;U),\label{privacyex}\\
	\text{subject to:}\ \ &X-Y-U,\label{Markovex}\\
	& \left\lVert[\sqrt{P_X}^{-1}](P_{X|U=u}-P_X) \right\rVert^2\leq \epsilon^2,\ \forall u\in\mathcal{U}.\label{localex}
	\end{align}
	\end{subequations}
Since the strong privacy criterion in \eqref{localex} implies the privacy criterion in \cite[Definition~4]{Calmon2} and the objective functions are the same, we conclude that the problem defined in \cite[Definition~4]{Calmon2} is an upper bound to \eqref{privacyex}. Furthermore, the upper bound in \cite{Calmon2} is equal to $\frac{1}{\lambda_{\text{min}}(X;Y)}\epsilon^2$ for $\epsilon^2\leq \lambda_{\text{min}}(X;Y)$, where $\sqrt{\lambda_{\text{min}}(X;Y)}$ is the minimum singular value of $Q_{X,Y}$, which is defined in \cite[Definition~2]{Calmon2}. Next, we show that $\frac{1}{\lambda_{\text{min}}(X;Y)}=\sigma^2_{\max}(W)$. The relation between $W$ and $Q_{X,Y}$ is as follows
\begin{align*}
Q_{X,Y}&=[\sqrt{P_X}^{-1}]P_{X,Y}[\sqrt{P_Y}^{-1}]=[\sqrt{P_X}^{-1}]P_{X|Y}[\sqrt{P_Y}]\\&=W^{-1}.
\end{align*}
Thus, we have $\frac{1}{\lambda_{\text{min}}(X;Y)}=\sigma^2_{\max}(W)$. Also, $\epsilon<\frac{|\sigma_{\text{min}}(P_{X|Y})|\min_{y\in\mathcal{Y}}P_Y(y)}{\sqrt{\max_{x\in{\mathcal{X}}}P_X(x)}}$ leads to the first region ($\epsilon^2\leq \lambda_{\text{min}}(X;Y)$) of the upper bound, since we have
\begin{align*}
|||W|||\leq \frac{1}{\min P_Y}|||P_{X|Y}^{-1}|||(\max P_X),
\end{align*}
which implies
\begin{align*}
\frac{(\sigma_{\text{min}}(P_{X|Y}))^2(\min P_Y)^2}{\max P_X}&\leq \frac{(\sigma_{\text{min}}(P_{X|Y}))^2\min P_Y}{\max P_X}\\
&\leq \frac{1}{\sigma_{\max}^2(W)}=\lambda_{\text{min}}(X;Y),
\end{align*}
where we used spectral norm defined as $|||A|||=\max_{||x||_2=1}||Ax||_2^2$, also $\max P_X=\max_{x\in{\mathcal{X}}}P_X(x)$ and $\min P_Y=\min_{y\in\mathcal{Y}} P_Y(y)$. The privacy mechanism found in this paper achieves $\sigma_{\max}^2(W)\epsilon^2$ for \eqref{privacyex}, and since \eqref{privacyex} is a lower bound to the problem defined in \cite[Definition~4]{Calmon2} and achieves the upper bound in \cite[Th.2]{Calmon2} for small $\epsilon$, we can conclude the upper bound in \cite[Th.2]{Calmon2} is tight for all $\epsilon<\frac{|\sigma_{\text{min}}(P_{X|Y})|\min_{y\in\mathcal{Y}}P_Y(y)}{\sqrt{\max_{x\in{\mathcal{X}}}P_X(x)}}$.
\end{proof}

	In next section, we study two extensions, where the idea of information geometry approximation is used.

	\section{Extensions}\label{exten}
	In this section, two problems are introduced. First, a fixed binary channel between the user and the adversary is considered and the agent is trying to find a mechanism to produce a binary random variable $U$, which maximizes $I(U;Y)$ under the Markov chain $X-Y-U-U'$ and privacy criterion on $X$ and $U'$. In the other scenario, $U$ is the stored data which can be used later, i.e., for future purposes, and $U'$ is the first instance that is disclosed with a standard post-processing mechanism.
		In second extension, the agent looks for a mechanism which finds $U$ based on observing $X$, maximizing the mutual information between $U$ and $Y$ while satisfying the privacy criterion on $U$ and $Z$ under the Markov chain $(Z,Y)-X-U$. In these extensions, small enough $\epsilon$ stands for all $\epsilon$ such that the second Taylor expansion can be used.
	\subsection{Privacy problem with sophisticated adversary}
	Similar to the previous problem let $P_{XY}$ denote the joint distribution of discrete random variables $(X,Y)$ and the leakage matrix defined by $P_{X|Y}$ be invertible. Similarly, let $X$ and $Y$ denote the private and the useful data with equal cardinality, i.e, $|\cal{X}|=|\cal{Y}|=\mathcal{K}$. Other considerations on $(X,Y)$ mentioned in section \ref{sec:system}, are assumed in this problem. Here, we add an invertible fixed binary channel between the user and an adversary denoted by $P_{U|U'}$ on $\mathbb{R}^{2\times 2}$, where we assume $|U|=|U'|=2$. $U'$ is the message received by the adversary and $U$ is the message received by the user. The agent tries to find a mechanism to produce $U$ such that maximizes $I(U;Y)$ while satisfying privacy criterion on $X$ and $U'$ under the Markov chain $X-Y-U-U'$. The privacy criterion employed in this problem is as follows
	\begin{align*}
		\left\lVert[\sqrt{P_X}^{-1}](P_{X|U'=u'}-P_{X})\right\rVert^2\leq \epsilon^2,\ u'\in \{u'_0,u'_1\}.
	\end{align*} 
	The information theoretic privacy problem can be characterized as follows
	\begin{subequations}
	\begin{align}
		\max_{P_{U|Y}} \ \ &I(U;Y),\label{privacy1}\\
		\text{subject to:}\ &X-Y-U-U',\\
		&\left\lVert[\sqrt{P_X}^{-1}](P_{X|U'=u'}-\!P_X)\right\rVert^2\! \!\!\leq \frac{1}{2}\epsilon^2\!,\  u' \!\! \in\!\! \{\!u'_0,u'_1\!\}.\label{local1}
	\end{align}
\end{subequations}
	Same as before, we assume that $\epsilon$ is a small quantity. 
	We define the matrix $P_{U|U'}\in\mathbb{R}^{2\times2}$ by $\begin{bmatrix}
	x &y\\z &t
	\end{bmatrix}$, where $x+z=1,\ y+t=1,$ and all $x$, $y$, $z$ and $t$ are non-negative. Furthermore, we show $P_{U|U'}^{-1}$ by $\begin{bmatrix}
	a &c\\b &d
	\end{bmatrix}$, where $a=\frac{t}{xt-zy}$, $b=\frac{-z}{xt-zy}$, $c=\frac{-y}{xt-zy}$ and $d=\frac{x}{xt-zy}$. 
	\begin{proposition}\label{tir}
			The tuple $(a,b,c,d)$ belongs to one of the following sets 
			\begin{align*}
	&A_1 = \\ &\{(a,b,c,d)|a\leq0, d\leq0, b\geq1, c\geq1,\ a+b\!=\!1,\ c+d\!=\!1\},\\
	&A_2 = \\ &\{(a,b,c,d)|a\geq1, d\geq1, b\leq0, c\leq0,\ a+b\!=\!1,\ c+d\!=\!1\}.
			\end{align*}
	\end{proposition}
\begin{proof}
	Since $x+z=1$ and $y+t=1$, we have
	\begin{align*}
	a+b=\frac{t-z}{xt-zy}=\frac{t-z}{(1-z)t-z(1-t)}=1.
	\end{align*}
	Since $t\geq0$ and $z\geq0$, one of $a$ and $b$ is non-negative and the other one is non-positive. Furthermore, since $a+b=1$, we have $a\leq0,\ b\geq1$ or $a\geq1,\ b\leq0$. Same proof can be used for $c$ and $d$. 
\end{proof}

	By using \eqref{local1}, we can write $P_{X|U=u}=P_X+\epsilon\cdot J_{u'}$, where $J_{u'}\in\mathbb{R}^{\mathcal{K}}$ is the perturbation vector that has three properties as follows 
	\begin{align}
	&\sum_{x\in\mathcal{X}} J_{u'}(x)=0,\ u'\in\{u'_0,u'_1\},\label{prop111}\\
	&P_{U'}(u'=u'_0)J_0+P_{U'}(u'=u'_1)J_1=\bm{0}\label{prop222},\\
	&\sum_{x\in\mathcal{X}}\frac{J_{u'}^2(x)}{P_X(x)}\leq 1,\ u'\in\{u'_0,u'_1\} \label{prop333},
	\end{align}
	where $\bm{0}\in\mathbb{R}^{\mathcal{K}}$.

	 Similarly, we can show that by using the concept of Euclidean Information theory, \eqref{local1} results in a leakage constraint.
	 \begin{proposition}
	 	For a small enough $\epsilon$, \eqref{local} results in a leakage constraint as follows
	 		\begin{align*}
	 	I(U';X)\leq \frac{1}{2}\epsilon^2+o(\epsilon^2).
	 	\end{align*}
	 \end{proposition}
\begin{proof}
	The proof is similar to Proposition~\ref{propp1}.
\end{proof}

	Next proposition shows that the post-processing inequality holds for the $\chi^2_{\text{information}}$ privacy constraint.
	\begin{proposition}
		Let $X-U-U'$ form a Markov chain. Then, we have
		\normalfont
		\begin{align*}
		\chi^2_{\text{information}}(X;U)&=\sum_u P_U(u) \chi^2(P_{X|U=u}||P_Y)\\&\geq \sum_{u'} P_{U'}(u') \chi^2(P_{X|U'=u'}||P_Y)\\&=\chi^2_{\text{information}}(X;U'). 
		\end{align*}
	\end{proposition}
	\begin{proof}
		For all $X\in\cal{X}$, let $P_{X|U'=u',U=u}$ defined on $\mathbb{R}^{|\cal X|}$ correspond to a distribution vector where each element equals to $P_{X=x|U'=u',U=u}$. We have
		\begin{align*}
		&\sum_{u} P_{U}(u) \chi^2(P_{X|U=u}||P_X)\\&=\sum_{u} P_{U}(u) \left\lVert[\sqrt{P_X}^{-1}](P_{X|U=u}-P_X) \right\rVert^2 \\
		&\stackrel{(a)}{=}\sum_{u',u} P_{U',U}(u',u) \left\lVert[\sqrt{P_X}^{-1}](P_{X|U'=u',U=u}-P_X) \right\rVert^2\\
		&=\sum_{u'}P_{U'}(u')\times\\&\sum_{u}P_{U|U'}(u|u')\left\lVert[\sqrt{P_X}^{-1}](P_{X|U'=u',U=u}-P_X) \right\rVert^2\\
		&\stackrel{(b)}{\geq} \sum_{u'}P_{U'}(u')\times\\&\left\lVert[\sqrt{P_X}^{-1}](\sum_{u}P_{U|U'}(u|u')P_{X|U'=u',U=u}-P_X) \right\rVert^2\\
		&=\sum_{u'}P_{U'}(u') \left\lVert[\sqrt{P_X}^{-1}](P_{X|U'=u'}-P_X) \right\rVert^2,
		\end{align*}
		where in step (a) we used the fact that $X-U-U'$ form a Markov chain. Furthermore, (b) follows from the fact that square of $\ell_2$-norm is convex, i.e., $f(x)=\lVert x\rVert_2^2$, is a convex function. A similar proof has been done in \cite[Th. 3]{Total}, where total variation has been used for privacy constraint. 
	\end{proof}
	\begin{remark}
		It is easy to check that the post-processing inequality does not hold for the strong $\chi^2$-privacy criterion in general, which is a per letter criterion, however, it holds on average ($\chi^2_{\normalfont\text{information}}$) as shown in the previous proposition. 
	\end{remark}
	Next proposition shows that the strong $\chi^2$-privacy constraint leads to a similar constraint on $X$ and $U$.
	\begin{proposition}
		Equation \eqref{local1} imposes a privacy constraint on $X$ and $U$ as follows
		\begin{align*}
		\left\lVert[\sqrt{P_X}^{-1}](P_{X|U=u_0}-P_X)\right\rVert^2\leq \epsilon^2(a^2+b^2),\\
		\left\lVert[\sqrt{P_X}^{-1}](P_{X|U=u_1}-P_X)\right\rVert^2\leq \epsilon^2(c^2+d^2).
		\end{align*}
	\end{proposition}
	\begin{proof}
		By using \eqref{joon} we have $P_{X|u_0}=aP_{X|u'_0}+bP_{X|u'_1}=P_X+\epsilon[aJ_{u'_0}+bJ_{u'_1}]$ and $P_{X|u_1}=P_X+\epsilon[cJ_{u'_0}+dJ_{u'_1}]$. Thus, we obtain
		\begin{align*}
		\left\lVert[\sqrt{P_X}^{-1}](P_{X|U=u_0}-P_X)\right\rVert^2&\!=\!\left\lVert[\sqrt{P_X}^{-1}](aJ_{u'_0}+bJ_{u'_1})\right\rVert^2\\&\stackrel{(a)}{\leq} \epsilon^2(a^2+b^2),
		\end{align*}
		where step (a) can be shown as follows 
		\begin{align*}
		&\left\lVert[\sqrt{P_X}^{-1}](aJ_{u'_0}+bJ_{u'_1})\right\rVert^2 \\&= \sum_x \left(a\left(([\sqrt{P_X}^{-1}]J_{u'_0})(x)\right)+b\left(([\sqrt{P_X}^{-1}]J_{u'_1})(x)\right)\right)^2\\
		\\&\stackrel{(b)}{\leq}\! \sum_x (\!a^2+b^2)\!\!\left(\!\!\left(\!([\sqrt{P_X}^{-1}]J_{u'_0}\!)(x)\!\right)^2\!\!\!\!\!+\!\!\left(\!([\sqrt{P_X}^{-1}]J_{u'_1}\!)(x)\!\right)^2\right)\\
		\\&= (a^2+b^2)\left(\left\lVert[\sqrt{P_X}^{-1}]J_{u'_0}\right\rVert^2+\left\lVert[\sqrt{P_X}^{-1}]J_{u'_1}\right\rVert^2 \right)\\
		&\stackrel{(c)}{\leq} \epsilon^2(a^2+b^2),
		\end{align*}
		where (b) comes from Cauchy Schwarz inequality and in step (c) we used \eqref{prop333}. A similar proof applies to $U=u_1$. Furthermore, by letting $J_{u_0}=aJ_{u'_0}+bJ_{u'_1}$ and $J_{u_1}=cJ_{u'_0}+dJ_{u'_1}$ we have the following properties using \eqref{prop111} and \eqref{prop222} for $J_u$
		\begin{align*}
		&\sum_{x\in\mathcal{X}} J_{u}(x)=0,\ u\in\{u_0,u_1\},\\
		&P_{U}(u=0)J_{u_0}+P_{U}(u=1)J_{u_1}=\bm{0},
		\end{align*}
		which means $J_{u_0}$ and $J_{u_1}$ are the perturbation vectors for the conditional distributions $P_{X|u_0}$ and $P_{X|u_1}$.
\end{proof}
	We show that $P_{Y|U=u}$ can be written as a linear perturbation of $P_Y$. Since the Markov chain $X-Y-U-U'$ holds, we can write
	\begin{align*}
	P_{X|u'_0}=P_{X|U}P_{U|u'_0},\
	P_{X|u'_1}=P_{X|U}P_{U|u'_1}.
	\end{align*}
	Thus, $P_{X|U'}=P_{X|U}P_{U|U'}$ and since $P_{U|U'}$ is invertible, we obtain
	\begin{align}\label{joon}
	[P_{X|u_0}\ \ P_{X|u_1}]=P_{X|U'}\begin{bmatrix}
	a & c\\
	b & d
	\end{bmatrix}.
	\end{align}
	Furthermore, by using the Markov chain we have $P_{Y|U=u} = P_{X|Y}^{-1}P_{X|U=u}$, which results in
	\begin{align*}
	&P_{Y|u_0}\!=\!P_{X|Y}^{-1}[aP_{X|u'_{0}}\!+\!bP_{X|u'_{1}}],\\
	&P_{Y|u_1}\!=\!P_{X|Y}^{-1}[cP_{X|u'_{0}}\!+\!dP_{X|u'_{1}}].
	\end{align*}
	 Considering $P_{Y|U=u_0}$, we have
	\begin{align}
	P_{Y|u_0}-P_Y&=P_{X|Y}^{-1}[aP_{X|u'_{0}}+bP_{X|u'_{1}}-P_X]\nonumber\\&=P_{X|Y}^{-1}[a(P_{X|u'_{0}}-P_X)+b(P_{X|u'_{1}}-P_X)]\nonumber\\
	P_{Y|u_0}&=P_Y+a\epsilon P_{X|Y}^{-1}J_{u'_0}+b\epsilon P_{X|Y}^{-1}J_{u'_1}.\label{y1}
	\end{align}
	Similarly, $P_{Y|U=u_1}$ is found as follows
	\begin{align}
	P_{Y|u_1}&=P_Y+c\epsilon P_{X|Y}^{-1}J_{u'_0}+d\epsilon P_{X|Y}^{-1}J_{u'_1}.\label{y2}
	\end{align}
	Now we can approximate $I(U;Y)$ by a squared Euclidean metric.
	\begin{proposition}
		For a small enough $\epsilon$, $I(U;Y)$ can be approximated as follows
		\begin{align*}
		&I(U;Y)\cong \\&\frac{1}{2}\epsilon^2\left[P_{u_0}\left\lVert W(aL_{u'_0}+bL_{u'_1})\right\rVert^2+P_{u_1}\left\lVert W(cL_{u'_0}+dL_{u'_1})\right\rVert^2\right],
		\end{align*}
		where $W$ is defined in Corollary~1 and $L_{u'}=[\sqrt{P_X}^{-1}]J_{u'}\in\mathbb{R}^{\mathcal{K}}$ for $ u'\in\{u'_0,u'_1\}$.
	\end{proposition}
	\begin{proof}
	By using \eqref{y1} and \eqref{y2} we have
	\begin{align*}
	I(Y;U)&=\sum_u P_U(u)D(P_{Y|U=u}||P_Y)\\
	&=P_{u_0}\!\!\sum_y\!\!P_{Y|u_0}\!\log(\frac{P_{Y|u_0}}{P_Y})\! +\!\!P_{u_1}\!\!\sum_y\!\!P_{Y|{u_1}}\!\log(\frac{P_{Y|u_1}}{P_Y})\\
	&=P_{u_0}\sum_y(P_Y+a\epsilon P_{X|Y}^{-1}J_{u'_0}+b\epsilon P_{X|Y}^{-1}J_{u'_1})\\&\times\log(1+\frac{\epsilon P_{X|Y}^{-1}(aJ_{u'_0}+bJ_{u'_1})}{P_Y})
	\\&+P_{u_1}\sum_y(P_Y+c\epsilon P_{X|Y}^{-1}J_{u'_0}+d\epsilon P_{X|Y}^{-1}J_{u'_1})\\&\times\log(1+\frac{\epsilon P_{X|Y}^{-1}(cJ_{u'_0}+dJ_{u'_1})}{P_Y})\\
	&=\frac{1}{2}\epsilon^2P_{u_0}\sum_y\frac{(P_{X|Y}^{-1}(aJ_{u'_0}+bJ_{u'_1}))^2}{P_Y}\\&+\frac{1}{2}\epsilon^2P_{u_1}\sum_y\frac{(P_{X|Y}^{-1}(cJ_{u'_0}+dJ_{u'_1}))^2}{P_Y}+o(\epsilon^2)\\
	&=\frac{1}{2}\epsilon^2P_{u_0}\left\lVert[\sqrt{P_Y}^{-1}]P_{X|Y}^{-1}(aJ_{u'_0}+bJ_{u'_1})\right\rVert^2\\&+\frac{1}{2}\epsilon^2P_{u_1}\left\lVert[\sqrt{P_Y}^{-1}]P_{X|Y}^{-1}(cJ_{u'_0}+dJ_{u'_1})\right\rVert^2\!\!+o(\epsilon^2)\\
	&\cong\frac{1}{2}\epsilon^2P_{u_0}\left\lVert[\sqrt{P_Y}^{-1}]P_{X|Y}^{-1}[\sqrt{P_X}](aL_{u'_0}+bL_{u'_1})\right\rVert^2\\
	&+\frac{1}{2}\epsilon^2P_{u_1}\left\lVert[\sqrt{P_Y}^{-1}]P_{X|Y}^{-1}[\sqrt{P_X}](cL_{u'_0}+dL_{u'_1})\right\rVert^2\\
	&=\frac{1}{2}\epsilon^2 P_{u_0}\left\lVert W(aL_{u'_0}+bL_{u'_1})\right\rVert^2\\&+\frac{1}{2}\epsilon^2P_{u_1}\left\lVert W(cL_{u'_0}+dL_{u'_1})\right\rVert^2.
	\end{align*}
	\end{proof}
	By locally approximating $I(U;Y)$, the main privacy problem in \eqref{privacy1} can be reduced to a simple quadratic problem.
	Substituting $L_u$ in \eqref{prop111}, \eqref{prop111} and \eqref{prop333} leads to the following result.
	\begin{corollary}
		For a small enough $\epsilon$, the privacy mechanism design problem in \eqref{privacy1} can be approximately solved by the following linear problem
		\begin{align}
		\max_{\{L_{u'},P_u\}} &P_{u_0}\!\left\lVert W(aL_{u'_0}\!+\!bL_{u'_1})\right\rVert^2\!\!+\!P_{u_1}\!\left\lVert\label{newprob22} W(cL_{u'_0}\!+\!dL_{u'_1})\right\rVert^2\\
		\text{subject to:}\ &\|L_{u'}\|^2\leq 1,\ u'\in\{u'_0,u'_1\},\\
		&\sum_x \sqrt{P_X(x)}L_{u'}(x)=0,\ u'\in\{u'_0,u'_1\},\label{orth11}\\
		&P_{u'_0}L_{u'_0}+P_{u'_1}L_{u'_1}=\bm{0},\label{orth222}
		\end{align}
		where $\bm{0}\in\mathbb{R}^{\mathcal{K}}$.
	\end{corollary}
	\begin{remark}
		\normalfont
		Condition \eqref{orth11} can be rewritten as $L_{u'_0}\perp \sqrt{P_X}$ and $L_{u'_1}\perp \sqrt{P_X}$. Also the maximization is over $\{L_{u'_0},L_{u'_1},P_{u_0},P_{u_1}\}$. $P_{u'_0}$ and $P_{u'_1}$ are replaced by $aP_{u_0}+cP_{u_1}$ and $bP_{u_0}+dP_{u_1}$, since we have $\begin{bmatrix}
		P_{u'_0}\\P_{u'_1}
		\end{bmatrix}=P_{U|U'}^{-1}\begin{bmatrix}
		P_{u_0}\\P_{u_1}
		\end{bmatrix}$.
		\end{remark}
	In next proposition we derive the solution of \eqref{newprob22}.
	\begin{proposition}
		The solution of \eqref{newprob22} is as follows
		\begin{align*}
		&L_{u'_0}=-L_{u'_1}=\psi,\\
		 &P_{u_0}=\frac{c-\frac{1}{2}}{c-a},\ P_{u_1}=\frac{\frac{1}{2}-a}{c-a},\ P_{u'_0}=P_{u'_0}=\frac{1}{2}\\
		 & \text{Maximum value}= 4(c-\frac{1}{2})(\frac{1}{2}-a)\sigma^2,
		\end{align*}
		where $\sigma^2$ is the largest singular value of $W$ with corresponding singular vector $\psi$.
	\end{proposition}
	\begin{proof}
		The proof is provided in Appendix D.
	\end{proof}
	\begin{corollary}
		The maximum value in \eqref{privacy1} can be approximated by $2\epsilon^2\sigma^2(c-\frac{1}{2})(\frac{1}{2}-a)$ for small $\epsilon$ and can be achieved by conditional distributions as follows
		\begin{align*}
		P_{Y|u_0}=P_Y+\epsilon(a-b)P_{X|Y}^{-1}[\sqrt{P_X}]\psi,\\
		P_{Y|u_1}=P_Y+\epsilon(c-d)P_{X|Y}^{-1}[\sqrt{P_X}]\psi,
		\end{align*}
		where $\sigma^2$ is the largest singular value of $W$ with corresponding singular vector $\psi$. Furthermore, the distribution of $U$ is as follows
		\begin{align*}
		P_{u_0}=\frac{c-\frac{1}{2}}{c-a},\ P_{u_1}=\frac{\frac{1}{2}-a}{c-a}.
		\end{align*}
	\end{corollary}
	\subsection{Privacy problem with utility provider}
	In this part, we consider a similar framework as in \cite{gun}, where we have an agent and a utility provider. The agent observes useful data denoted by RV $X$ and the utility provider is interested in target data denoted by RV $Y$ which is not directly accessible by the agent but correlated with RV $X$. The agent receives utility by disclosing information about $Y$. Furthermore, we assume $X$ is dependent on the private data denoted by RV $Z$, and tried to keep it private and not disclose much information about $Z$. Thus, the agent uses a privacy mechanism to produce $U$ and tries to maximize the utility measured by $I(U;Y)$ and at the same time satisfies the privacy criterion. RV $U$ denotes the disclosed data. Here we assume that all random variables are discrete and have finite support, i.e., $|\mathcal{X}|,|\mathcal{Y}|,|\mathcal{Z}|<\infty$. Since the disclosed data is produced by observing $X$ and the variables $X$, $Y$ and $Z$ are correlated, we have the Markov chain $(Z,Y)-X-U$. We assume that $|\mathcal{X}|=|\mathcal{Z}|=\mathcal{K}$ and the leakage matrix $P_{Z|X}\in\mathbb{R}^{\mathcal{K}\times\mathcal{K}}$ is invertible. Furthermore, the marginal vectors $P_X,\ P_Z$ and $P_Y$ contain non-zero elements. Here, privacy is measured as follows  
	\begin{align*}
\left\lVert[\sqrt{P_Z}^{-1}](P_{Z|U=u}-P_{Z})\right\rVert^2\leq \epsilon^2,\ \forall u\in \mathcal{U}.
	\end{align*}
	The privacy problem is characterized as follows 
	\begin{subequations}
	\begin{align}
	\max_{P_{U|X}} \ \ &I(U;Y),\label{privacy13}\\
	\text{subject to:}\ &(Z,Y)-X-U,\\
	&\left\lVert[\sqrt{P_Z}^{-1}](P_{Z|U=u}-P_Z)\right\rVert^2\leq \frac{1}{2}\epsilon^2,\ \forall u\in\mathcal{U},\label{local13}
	\end{align}
	\end{subequations}
\begin{remark}
	\normalfont
	By using Fenchel-Eggleston-Carath\'{e}odory's Theorem \cite{el2011network}, it can be shown that it suffices to consider $U$ such that $|\mathcal{U}|\leq |\mathcal{X}|+1$. Furthermore, the maximum in \eqref{privacy13} is achieved so we used maximum instead of supremum. 
\end{remark}
	Similarly, \eqref{local13} results in $P_{Z|U=u}=P_Z+\epsilon J_u$, where $J_u\in\mathbb{R}^\mathcal{K}$ is the perturbation vector that has the following three properties 
	\begin{align}
	\sum_{z\in\mathcal{Z}} J_u(z)=0,\ \forall u,\label{prop13}\\
	\sum_{u\in\mathcal{U}} P_U(u)J_u=\bm{0}, \label{prop23}\\
	\sum_{z\in\mathcal{Z}}\frac{J_u^2(z)}{P_Z(z)}\leq 1, \forall u\label{prop33},
	\end{align}
	where $\bm{0}\in\mathbb{R}^\mathcal{K}$.
	By following the same procedure in Proposition~1 and using the Euclidean information concept, it can be shown that for a small enough $\epsilon$, \eqref{local13} results in the following leakage constraint
	\begin{align*}
	I(Z;U)\leq \frac{1}{2}\epsilon^2+o(\epsilon^2).
	\end{align*}
	Now we show that $P_{Y|U=u}$ can be written as a linear perturbation of $P_Y$. Since the Markov chain $(Z,Y)-X-U$ holds, we can write $P_{X|U=u}=P_{Z|X}^{-1}P_{Z|U=u}$. Thus,
	\begin{align*}
	P_{Y|U=u}=P_{Y|X}P_{X|U=u} = P_{Y|X}P_{Z|X}^{-1}P_{Z|U=u}.
	\end{align*}
	Using $P_{Z|U=u}=P_Z+\epsilon\cdot J_u$, $P_{Y|U=u}$ can be written as follows
	\begin{align*}
	P_{Y|U=u}= P_{Y|X}P_{Z|X}^{-1}(P_Z+\epsilon\cdot J_u)=P_Y+\epsilon\cdot P_{Y|X}P_{Z|X}^{-1}J_u.
	\end{align*}
	The next proposition shows that $I(U;Y)$ can be locally approximated by a squared Euclidean metric.
	\begin{proposition}
		For a small enough $\epsilon$, $I(U;Y)$ can be approximated as follows
		\begin{align}
		I(U;Y)\cong\frac{1}{2}\epsilon^2\sum_u P_U(u)\left\lVert [\sqrt{P_Y}^{-1}]P_{Y|X}P_{Z|X}^{-1}[\sqrt{P_Z}]L_u \right\rVert^2, \label{approx33}
		\end{align}
		where $L_u=[\sqrt{P_Z}^{-1}]J_u$.
	\end{proposition}
	\begin{proof}
		By using the local approximation of the KL-divergence we have
		\begin{align*}
		I(Y;U)&=\sum_u P_U(u)D(P_{Y|U=u}||P_Y)\\&=\sum_u P_U(u)\sum_y P_{Y|U=u}(y)\log\left(\frac{P_{Y|U=u}(y)}{P_Y(y)}\right)\\&=\sum_u \!P_U\!\sum_y\! P_{Y|U=u}(y)\log\!\left(\!\!1\!+\!\epsilon\frac{P_{Y|X}P_{Z|X}^{-1}J_u(y)}{P_Y(y)}\!\right)\\&=\frac{1}{2}\epsilon^2\sum_u P_U\sum_y
		\frac{(P_{Y|X}P_{Z|X}^{-1}J_u)^2}{P_Y}+o(\epsilon^2)\\
		&=\frac{1}{2}\epsilon^2\sum_u P_U\|[\sqrt{P_Y}^{-1}]P_{Y|X}P_{Z|X}^{-1}J_u\|^2+o(\epsilon^2)
		\\&\cong\frac{1}{2}\epsilon^2\sum_u P_U\|[\sqrt{P_Y}^{-1}]P_{Y|X}P_{Z|X}^{-1}[\sqrt{P_Z}]L_u\|^2.
		\end{align*}
		\end{proof}
	By substituting $L_u$ in \eqref{prop13}, \eqref{prop23} and \eqref{prop33}, and using the local approximation in \eqref{approx33} we obtain the following result.  
	\begin{corollary}
		For a small enough $\epsilon$, the privacy mechanism design problem in \eqref{privacy13} can be approximately solved by the following linear problem
		\begin{align}
		\max_{\{L_u,P_U\}} \ &\sum_u P_U(u)\|W_1W_2L_u\|^2,\label{newprob23}\\
		\text{subject to:}\ &\|L_u\|^2\leq 1,\ \forall u\in\mathcal{U},\\
		&\sqrt{P_Z}\perp L_u,\ \forall u,\label{orth3}\\
		&\sum_u P_U(u)L_u=\bm{0},\label{orth23}
		\end{align}
		where $W_1=[\sqrt{P_Y}^{-1}]P_{Y|X}[\sqrt{P_X}]$ and $W_2=[\sqrt{P_X}^{-1}]P_{Z|X}^{-1}[\sqrt{P_Z}]$.
	\end{corollary}
Similar to the Proposition~3, without loss of optimality we can choose $U$ as a uniform binary RV. Thus, \eqref{newprob23} reduces to the following problem
	\begin{align}
	\max_{L:L\perp \sqrt{P_Z},\ \|L\|^2\leq 1} \|W_1W_2\cdot L\|^2.\label{max33}
	\end{align}
	Let $L^*$ maximizes \eqref{max33}, thus, the conditional distributions $P_{Y|U=u}$ which maximizes \eqref{newprob23} are given by  
	\begin{align}
	P_{Y|U=0}=P_Y+\epsilon P_{Y|X}P_{Z|X}^{-1}[\sqrt{P_Z}]L^*,\label{condis31}\\
	P_{Y|U=1}=P_Y-\epsilon P_{Y|X}P_{Z|X}^{-1}[\sqrt{P_Z}]L^*.\label{condis32}
	\end{align}
	In the next theorem, the solution of \eqref{max33} is derived.
	\begin{theorem}
		Let $\sigma_{\max}$ be the largest singular value of $W_1W_2$ corresponding to the singular vector $\psi$. Furthermore, let $\phi$ be the singular vector of $W_1W_2$ corresponding the second largest singular value. If $\sigma_{\max}>1$, $\psi$ maximizes \eqref{max33}, and if $\sigma_{\max}=1$, $\phi$ is the maxmizer of \eqref{max33}.
	\end{theorem}
	\begin{proof}
		The largest singular value of $W_1$ is $1$ corresponding to singular vector $\sqrt{P_Z}$ and the smallest singular value of $W_2$ is $1$ corresponding to singular vector $\sqrt{P_Z}$. Furthermore, we show that $1$ is one of the singular values of $W_1W_2$ corresponding to singular vector $\sqrt{P_Z}$. We have
		\begin{align*}
		&W_2^TW_1^TW_1W_2\sqrt{P_Z}=\\
		&[\sqrt{P_Z}]^T\left(P_{Z|X}^{-1}\right)^T[\sqrt{P_X}^{-1}]^T[\sqrt{P_X}]^TP_{Y|X}^T[\sqrt{P_Y}^{-1}]^T\times\\
		&[\sqrt{P_Y}^{-1}]P_{Y|X}[\sqrt{P_X}][\sqrt{P_X}^{-1}]P_{Z|X}^{-1}[\sqrt{P_Z}]\sqrt{P_Z}=\\
		&[\sqrt{P_Z}]^T\left(P_{Z|X}^{-1}\right)^T[\sqrt{P_X}^{-1}]^T[\sqrt{P_X}]^TP_{Y|X}^T\bm{1}=\\
		&[\sqrt{P_Z}]^T\bm{1}=\sqrt{P_Z}.
		\end{align*} 
		Thus, we have two cases as $\sigma_{\max}>1$ and $\sigma_{\max}=1$. In first case, $\psi$ is orthogonal to $\sqrt{P_Z}$ and so maximizes \eqref{max33}. In second case, $\psi=\sqrt{P_Z}$ and $\phi$ is orthogonal to $\sqrt{P_Z}$. Thus, $\phi$ maximizes $\eqref{max33}$. There are no other cases since $1$ is one of the singular values.
	\end{proof}
	\begin{corollary}
		Let $\sigma_{\max}$ and $\sigma_2$ be the first and second largest singular values of $W_1W_2$. If $\sigma_{\max}>1$, the maximum value in \eqref{privacy13} can be approximated by $\frac{1}{2}\epsilon^2\sigma_{\max}^2$ and can be achieved by a privacy mechanism characterized by conditional distributions found in \eqref{condis31} and \eqref{condis32} where $L^*=\psi$. Otherwise, the maximum value can be approximated by $\frac{1}{2}\epsilon^2\sigma_{2}^2$ and can be achieved by \eqref{condis31} and \eqref{condis32} where $L^*=\phi$. 
	\end{corollary}
	\begin{remark}
		\normalfont
		One simple example for the second case where $\sigma_{\max}=1$ is letting $P_{Z|X}=P_{Y|X}$. In this case, $W_1=W_2^{-1}$ and so all singular values of $W_1W_2$ are equal to one. The maximum value in \eqref{max33} is $1$ and can be achieved by any vector orthogonal to $\sqrt{P_Z}$.
	\end{remark}
    \begin{remark}
    	\normalfont
    	One sufficient condition for the first case where $\sigma_{\max}>1$, is to have $\sigma_{\max}(W_2)=\frac{1}{\sigma_{\min}(W_1)}$ and not all singular values are equal to $1$. Since in this case we have
    	\begin{align*}
    	||| W_1W_2|||&\geq \frac{|||W_2|||}{|||W_1^{-1}|||}=\frac{\sigma_{\max}(W_2)}{\sigma_{\max}(W_1^{-1})}\\
    	&=\sigma_{\max}(W_2)\sigma_{\min}(W_1)=1,
    	\end{align*}
    	where we used the spectral norm. 
    \end{remark}

\section{conclusion}\label{concul}
We have shown that Euclidean information theory can be used to linearize an information-theoretic disclosure control problem. When a small $\epsilon$ privacy leakage is allowed, a simple approximate solution is derived. A geometrical interpretation of the privacy mechanism design is provided. Four linear spaces are introduced to further interpret the structure of the optimization problem. In particular, we look for a vector satisfying the constraint of having the largest Euclidean norm in other space, leading to finding the largest principle singular value of a matrix. The proposed approach establishes a useful and general design framework, which has been demonstrated in two problem extensions that included an adversary and privacy design with utility provider. 
	\section*{Appendix A}
	As shown in \eqref{pyu}, $P_{Y|U=u}$ must belong to $\Psi$ for every $u\in\mathcal{U}$ which is defined as follows
	\begin{align*}
	\Psi=\{ y\in\mathbb{R}^{\mathcal{K}}| y=P_Y+\epsilon P_{X|Y}^{-1}J,\ \left\lVert J\right\rVert^2_{P_X}\leq 1,\ \bm{1}^T\cdot J=0\},
	\end{align*}
	where $\left\lVert J\right\rVert^2_{P_X}=\sum_x \frac{J(x)^2}{P_X(x)}$ is the weighted Euclidean norm.
	For all $\epsilon\!<\!\frac{|\sigma_{\text{min}}(P_{X|Y})|\min_{y\in\mathcal{Y}}P_Y(y)}{\sqrt{\max_{x\in{\mathcal{X}}}P_X(x)}}$, any point in $\Psi$ is a probability distribution and hence $\Psi$ is a subset of the standard $\mathcal{K}-1$ dimension simplex. Thus, $\Psi$ is bounded. Let $\mathcal{J}_1=\{J\in\mathbb{R}^\mathcal{K}|\bm{1}^T\cdot J=0\}$ and $\mathcal{J}_2=\{J\in\mathbb{R}^\mathcal{K}|\left\lVert J\right\rVert^2_{P_X}\leq 1\}$. $\mathcal{J}_1$ and $\mathcal{J}_2$ correspond to a hyperplane and an elipsoide, respectively. The set $\mathcal{J}=\mathcal{J}_1\cap\mathcal{J}_2$ is closed since each $\mathcal{J}_1$ and $\mathcal{J}_2$ is closed. Considering the sequence $\{y_0,y_1,..\}$ where each $y_i$ is inside the set $\Psi$, we have
	\begin{align*}
	\lim_{i\rightarrow \infty} y_i=\lim_{i\rightarrow \infty} P_Y+\epsilon P_{X|Y}^{-1}J_{i}=P_Y+\epsilon P_{X|Y}^{-1}\lim_{i\rightarrow \infty} J_{i}.
	\end{align*}
	Since $\mathcal{J}$ is a closed set $\lim_{i\rightarrow \infty} J_{i}\in\mathcal{J}$ and hence $\lim_{i\rightarrow \infty} y_i\in \Psi$. Thus, $\Psi$ is a compact set. We define a vector mapping $\theta:\Psi\rightarrow \mathbb{R}^{\cal K}$ as follows
		\begin{align*}
	\theta_i(p_{Y|U(\cdot|U)})&=p_{Y|U}(y_i|u),\ i\in [1:\mathcal{K}-1],\\
	\theta_{\mathcal{K}}&=H(Y|U=u).
	\end{align*}
	Since the mapping $\theta$ is continuous and the set $\Psi$ is compact, by using Fenchel-Eggleston-Carath\'{e}odory's Theorem \cite{el2011network} for every $U$ with p.m.f $F(u)$ there exists a random variable $U'$ with p.m.f $ F(u')$ such that $|\cal U'|\leq \cal K$ and collection of conditional p.m.fs $P_{Y|U'}(\cdot|u')\in \Psi$ where
	\begin{align*}
	\int_u \theta_i(p(y|u))dF(u)=\sum_{u'\in\cal U'}\theta_i(p(y|u'))p(u').
	\end{align*}
	It ensures that by replacing $U$ by $U'$, $I(U;Y)$ and the distribution $P_Y$ are preserved. Furthermore, the condition $\sum_{u'}P_{U'}(u')J_{u'}=\bm{0}$ is satisfied since we have
	\begin{align*}
	P_Y=\sum_{u'} P_{U'}P_{Y|U'=u'}\rightarrow P_X=\sum_{u'} P_{U'}P_{X|U'=u'}\\
	\sum_{u'} P_{U'}(P_{X|U'=u'}-P_X)=\bm{0}\rightarrow \sum_{u'}P_{U'}(u')J_{u'}=\bm{0}.
	\end{align*}
	Note that any point in $\Psi$ satisfies the strong privacy criterion, i.e., the equivalent $U'$ satisfies the per-letter privacy criterion as well. Thus, without loss of optimality we can assume $|\mathcal{U}|\leq \mathcal{K}$.
	
	Let $\mathcal{A}=\{P_{U|Y}(\cdot|\cdot)|U\in\mathcal{U}, Y\in\mathcal{Y},||\cal U|\leq \cal K \}$ and $\mathcal{ A}_y=\{P_{U|Y}(\cdot|y)|U\in\mathcal{U},|\cal U|\leq \cal K \}$, $\forall y\in \cal Y$. $\mathcal{A}_y$ is a standard $|\mathcal {U}|-1$ simplex and since $|\cal U|\leq |\cal Y|<\infty$ it is compact. Thus $\mathcal{A}=\cup_{y\in\mathcal{Y}}\mathcal{A}_y$ is compact. And the set $\mathcal{A}'=\{P_{U|Y}(\cdot|\cdot)\in\mathcal{A}|X-Y-U, \left\lVert[P_X]^{-1}(P_{X|U=u}-P_X)\right\rVert^2\leq \epsilon^2,\ \forall u\}$ is a closed subset of $\mathcal{A}$ since $\chi^2$ information is closed of the interval $[0,\epsilon^2]$. Therefore, $\mathcal {A}'$ is compact. Since $I(U;Y)$ is a continuous mapping over $\mathcal{A}'$, the supremum is achieved. Thus, we use maximum instead of supremum.
	
	\section*{Appendix B}\label{appa}
	The KL divergence is denoted by $D(\cdot||\cdot)$.
	\begin{align*}
		I(X;U)&=\sum_{u\in \mathcal{U}}P_U(u)D(P_{X|U=u}||P_X)\\&=\sum_u P_U(u)\sum_x P_{X|U=u}\log(\frac{P_{X|U=u}}{P_X})\\&\stackrel{(a)}{=}\sum_u P_U(u)\sum_x (P_X+\epsilon\cdot J_u)\log(1+\epsilon\frac{J_u}{P_X})\\
		&=\sum_u P_U(u)[\sum_x (\epsilon J_u+\frac{1}{2}\epsilon^2 \frac{J_u^2}{P_X})]+o(\epsilon^2)\\
		&=\frac{1}{2}\epsilon^2\sum_{u\in \mathcal{U}}P_U(u)\|[\sqrt{P_X}^{-1}]J_u\|^2+o(\epsilon^2),\\
		&\stackrel{(b)}{\leq}\frac{1}{2}\epsilon^2+o(\epsilon^2),
	\end{align*}
	where (a) follows from $P_{X|U=u}=P_X+\epsilon\cdot J_u$ and (b) follows from the third property of $J_u$ stated in \eqref{prop3}. Furthermore, for approximating $I(U;X)$ we should have $|\epsilon\frac{J_u(x)}{P_X(x)}|<1$ for all $x$ and $u$. One sufficient condition is to have $\epsilon<\frac{\min_{x\in\mathcal{X}}P_X(x)}{\sqrt{\max_{x\in\mathcal{X}}P_X(x)}}$.
	Thus the privacy criterion implies a bounded mutual information leakage.
	\section*{Appendix C}\label{appb}
	We first show that the smallest singular value of $W$ is $1$ with $\sqrt{P_X}$ as corresponding right singular vector.
	We have
	\begin{align*}
	&W^TW\sqrt{P_X}\\&=[\sqrt{P_X}](P_{X|Y}^T)^{-1}[\sqrt{P_Y}^{-1}][\sqrt{P_Y}^{-1}]P_{X|Y}^{-1}[\sqrt{P_X}]\sqrt{P_X}\\
	&=[\sqrt{P_X}](P_{X|Y}^T)^{-1}[\sqrt{P_Y}^{-1}][\sqrt{P_Y}^{-1}]P_Y\\
	&=[\sqrt{P_X}](P_{X|Y}^T)^{-1}\bm{1}=[\sqrt{P_X}]\bm{1}=\sqrt{P_X}.
	\end{align*}
	Now we show that all other singular values are greater than or equal to 1. Equivalently, we show that all singular values of $W^{-1}=[\sqrt{P_X}^{-1}]P_{X|Y}[\sqrt{P_Y}]$ are smaller than or equal to 1, i.e., we need to prove that for any vector $\alpha\in\mathbb{R}^{\mathcal{K}}$ we have 
	\begin{align}
	||W^{-1}\alpha||^2\leq||\alpha||^2.
	\end{align}
	In the following, we use $P_{Y_j}=P_Y(y_j)$, $P_{X_i}=P_X(x_i)$ and $P_{X_i|Y_j}=P_{X|Y}(x_i|y_j)$ for simplicity. 
	More explicitly we claim to have
	\begin{align*}
	&\alpha^T(W^{-1})^TW^{-1}\alpha=\sum_{j=1}^{\mathcal{K}}\alpha_j^2\sum_{i=1}^{\mathcal{K}}\frac{P_{X_i|Y_j}^2P_{Y_j}}{P_{X_i}}+\\&\sum_{\begin{array}{c} \substack{m,n=1\\ m\neq n} \end{array}}^{\mathcal{K}}\!\!\!\!\alpha_m\alpha_n\sum_{i=1}^{\mathcal{K}}\frac{P_{X_i|Y_m}P_{X_i|Y_n}\sqrt{P_{Y_n}P_{Y_m}}}{P_{X_i}}\leq\sum_{i=1}^{\mathcal{K}}\alpha_i^2.
	\end{align*}
	By using $\frac{P_{X_i|Y_j}^2P_{Y_j}}{P_{X_i}}=P_{X_i|Y_j}P_{Y_j|X_i}$, we can rewrite the last inequality as follows
	\begin{align*}
	&\sum_{j=1}^{\mathcal{K}}\alpha_j^2\sum_{i=1}^{\mathcal{K}}P_{X_i|Y_j}P_{Y_j|X_i}+\\&\sum_{\begin{array}{c}   \substack{m,n=1\\ m\neq n} \end{array}}^{\mathcal{K}}\!\!\!\!\alpha_m\alpha_n\sum_{i=1}^{\mathcal{K}}\frac{P_{X_i|Y_m}P_{X_i|Y_n}\sqrt{P_{Y_n}P_{Y_m}}}{P_{X_i}}\leq\sum_{i=1}^{\mathcal{K}}\alpha_i^2,
	\end{align*}
	Equivalently, by using $\sum_{i=1}^{\mathcal{K}}\sum_{m=1}^{\mathcal{K}}P_{X_i|Y_j}P_{Y_m|X_i}=1$, we claim to have 
	\begin{align*}
	&\sum_{\begin{array}{c}   \substack{m,n=1\\ m\neq n} \end{array}}^{\mathcal{K}}\!\!\!\!\alpha_m\alpha_n\sum_{i}\frac{P_{X_i|Y_m}P_{X_i|Y_n}\sqrt{P_{Y_n}P_{Y_m}}}{P_{X_i}}\leq\\ &\sum_j \alpha_j^2[\sum_i\sum_{m\neq j}P_{X_i|Y_j}P_{Y_m|X_i}].
	\end{align*}
	Finally, we can see that the last inequality holds, since for any $i$ by using the inequality of arithmetic and geometric means and  $P_{X_i|Y_m}P_{Y_n|X_i}P_{X_i|Y_n}P_{Y_m|X_i}=\frac{P_{X_i|Y_m}P_{X_i|Y_n}P_{X_i,Y_n}P_{X_i,Y_m}}{P_{X_i}^2}=\!\left(\!\frac{P_{X_i|Y_m}P_{X_i|Y_n}\sqrt{P_{Y_n}P_{Y_m}}}{P_{X_i}}\!\right)^2\!\!\!$, we have
	$
	2\alpha_m\alpha_n\frac{P_{X_i|Y_m}P_{X_i|Y_n}\sqrt{P_{Y_n}P_{Y_m}}}{P_{X_i}}\leq \alpha_m^2P_{X_i|Y_m}P_{Y_n|X_i}+\alpha_n^2P_{X_i|Y_n}P_{Y_m|X_i},
	$
	where we use 
	\begin{align*}
	P_{X_i|Y_m}P_{Y_n|X_i}P_{X_i|Y_n}P_{Y_m|X_i}&=\frac{P_{X_i|Y_m}P_{X_i|Y_n}P_{X_i,Y_n}P_{X_i,Y_m}}{P_{X_i}^2}\\&=\!\left(\!\frac{P_{X_i|Y_m}P_{X_i|Y_n}\sqrt{P_{Y_n}P_{Y_m}}}{P_{X_i}}\!\right)^2\!\!\!.
	\end{align*}
	
	Therefore, one is the smallest singular value of $W$ with $\sqrt{P_X}$ as corresponding right singular vector. Furthermore, we have that the right singular vector of the largest singular value is orthogonal to $\sqrt{P_X}$. Thus, the principal right-singular vector is the solution of \eqref{max2}.
	\section*{Appendix D}\label{appD}
	First, assume that the maximum occurs in non-zero $P_{u'_0}$ and $P_{u'_1}$. For simplicity we show $P_{u_0}$ and $P_{u_1}$ by $P_0$ and $P_1$, also we show $P_{u'_0}$ and $P_{u'_1}$ by $P'_0$ and $P'_1$.
	By using \eqref{orth222}, we have
	\begin{align*}
	L_{u'_0}=-\frac{P'_1}{P'_0}L_{u'_1}=-\frac{bP_0+dP_1}{aP_0+cP_1}L_{u'_1}.
	\end{align*}
	Since $\left\lVert L_{u'_0} \right\rVert^2\leq 1$, thus, $\left\lVert\frac{P'_1}{P'_0}L_{u'_1}\right\rVert^2\leq 1$, which results in $\left\lVert L_{u'_1}\right\rVert\leq \min \{1,(\frac{P'_0}{P'_1})^2\} \leq 1$. With the same argument $\left\lVert L_{u'_0}\right\rVert\leq \min \{1,(\frac{P'_1}{P'_0})^2\} \leq 1$. Now we consider two cases: \\
	1. Case 1: $|P'_0|\geq|P'_1|$, 2. Case 2: $|P'_1|\geq|P'_0|$.\\
	\textbf{Case 1} : 
	In this case we have $|aP_0+cP_1|\geq|bP_0+dP_1|$, which results in
	\begin{align}
	(a-c)P_0+c\geq\frac{1}{2},\label{jj}
	\end{align}
	since $b=1-a$ and $d=1-c$.
	Then, we substitute $L_{u'_0}$ by $-\frac{bP_0+dP_1}{aP_0+cP_1}L_{u'_1}$ in the objective function, which results in
	\begin{align*}
	&P_0\left\lVert W(aL_{u'_0}+bL_{u'_1})\right\rVert^2+P_1\left\lVert W(cL_{u'_0}+dL_{u'_1})\right\rVert^2\\
	&=\left\lVert WL_{u'_1} \right\rVert P_0\left(b-\frac{abP_0+adP_1}{aP_0+cP_1}\right)^2\\&+\left\lVert WL_{u'_1} \right\rVert P_1\left(b-\frac{cbP_0+cdP_1}{aP_0+cP_1}\right)^2\\
	&=\left\lVert WL_{u'_1} \right\rVert P_0\frac{(bc-ad)^2P_1^2}{(aP_0+cP_1)^2}
	+\left\lVert WL_{u'_1} \right\rVert P_1\frac{(ad-bc)^2P_0^2}{(aP_0+cP_1)^2}\\
	&=\left\lVert WL_{u'_1}\right\rVert(bc-ad)^2\left(\frac{P_0(1-P_0)}{((a-c)P_0+c)^2}\right)
	\end{align*}
	Now we show that the maximum of $f(P_0)=\frac{P_0(1-P_0)}{((a-c)P_0+c)^2}$ occurs in $P_0^*=\frac{c-\frac{1}{2}}{c-a}$. The derivative of $f$ with respect to $P_0$ is as follows
	\begin{align*}
	\frac{d}{dP_0}f=\frac{c-(a+c)P_0}{((a-c)P_0+c)^2}.
	\end{align*}
	By using Proposition~\ref{tir}, we have two cases for $a$ and $c$, $a\geq1,\ c\leq0$ or $a\leq0,\ c\geq1$. For $a\geq1,\ c\leq0$ we have $a-c\geq0$, which implies $P_0\geq\frac{\frac{1}{2}-c}{a-c}$ by using \eqref{jj}. We show that $f(P_0)$ is a decreasing function in this case. If $a+c\geq0$, then $c-(a+c)P_0\leq0$ and if $a+c\geq0$, then $-(a+c)P_0\leq-(a+c)$ which results in $c-(a+c)P_0\leq-a\leq -1<0$. Thus, for $a\geq1,\ c\leq0$, $f(P_0)$ is decreasing and its maximum happens in $P^*_0=\frac{\frac{1}{2}-c}{a-c}$. Now consider $a\leq0,\ c\geq1$. In this case we have $P_0\leq\frac{\frac{1}{2}-c}{a-c}$. We show that $f(P_0)$ is an increasing function. If $a+c\leq0$, then $P_0(a+c)\leq0$ which results in $c\geq1>0\geq(a+c)P_0$. And if $a+c\geq 0$, then $(a+c)P_0\leq\frac{(c-\frac{1}{2})(a+c)}{c-a}\leq c$, since $2ac\leq 0\leq \frac{a+c}{2}$. Thus, $f(P_0)$ is an increasing function and its maximum occurs in $P^*_0=\frac{\frac{1}{2}-c}{a-c}$. The maximum value of $f(P_0)$ is $\frac{(c-\frac{1}{2})(\frac{1}{2}-a)}{4(c-a)^2}$.\\
	\textbf{Case 2}: 
	In this case we have $|aP_0+cP_1|\leq|bP_0+dP_1|$, which results in
	\begin{align}
	(a-c)P_0+c\geq\frac{1}{2},\label{jjj}
	\end{align}
	We substitute $L_{u'_1}$ by $-\frac{aP_0+cP_1}{bP_0+dP_1}L_{u'_0}$ in the objective function, which results in
	\begin{align*}
	&P_0\left\lVert W(aL_{u'_0}+bL_{u'_1})\right\rVert^2+P_1\left\lVert W(cL_{u'_0}+dL_{u'_1})\right\rVert^2\\
	&=\left\lVert WL_{u'_0}\right\rVert(bc-ad)^2\left(\frac{P_0(1-P_0)}{((b-d)P_0+d)^2}\right)
	\end{align*}
	By the same arguments it can be shown that maximum of $\frac{P_0(1-P_0)}{((b-d)P_0+d)^2}$ occurs in $P^*_0=\frac{d-\frac{1}{2}}{d-b}=\frac{\frac{1}{2}-c}{a-c}$. 
	Thus for both cases we have
	\begin{align*}
	P^*_0=\frac{\frac{1}{2}-c}{a-c},\ p^*_1=\frac{a-\frac{1}{2}}{a-c},\
	P'^*_0=P'^*_1=\frac{1}{2}.
	\end{align*}
	So the maximum of \eqref{newprob22} occurs in $L_{u'_0}=-L_{u'_1}=\psi$, where $\psi$ is the singular vector corresponding to largest singular value of $W$, if both $P'_0$ and $P'_1$ are non-zero, and the maximum value is $4(c-\frac{1}{2})(\frac{1}{2}-a)\sigma^2$.
	
	Now we assume that $P'_0$ or $P'_1$ for instance $P'_0$ is zero, which implies $L_{u'_1}=0$ and $P'_1=1$. Thus, the objective function reduces to
	\begin{align*}
	\left\lVert WL_{u'_0}\right\rVert ^2\left( a^2P_0+c^2P_1\right).
	\end{align*} 
	Since $P'_0=aP_0+cP_1=0$, we have
	\begin{align*}
    P_0=-\frac{c}{a}P_1\rightarrow P_0=\frac{-c}{a-c},\ P_1=\frac{a}{a-c},
	\end{align*}
	Thus, the objective function is
	$\left\lVert WL_{u'_0}\right\rVert ^2(-ac)$,
	where the maximum value is $-ac\sigma^2$. We show that $4(c-\frac{1}{2})(\frac{1}{2}-a)\geq-ac$. This is true since we have $2(a+c)-3ac\geq1$ due to $a\geq1,\ c\leq0$ or $a\leq0,\ c\geq1$. Thus, the maximization of \eqref{newprob22} occurs in $P'_0=P'_1=\frac{1}{2}$. Furthermore, $L_{u'_0}=-L_{u'_1}=\psi$ satisfies the conditions \eqref{orth11} and \eqref{orth222}, since $\psi$ is orthogonal to $\sqrt{P_X}$.
	\section*{Acknowledgment}
	The work was partially supported by the Digital Futures research center and the Strategic Research Agenda Program, Information and Communication Technology - The Next Generation (SRA ICT - TNG), through the Swedish Government.
\bibliographystyle{IEEEtran}
\bibliography{IEEEabrv,IZS}

\begin{thebibliography}{10}
\providecommand{\url}[1]{#1}
\csname url@samestyle\endcsname
\providecommand{\newblock}{\relax}
\providecommand{\bibinfo}[2]{#2}
\providecommand{\BIBentrySTDinterwordspacing}{\spaceskip=0pt\relax}
\providecommand{\BIBentryALTinterwordstretchfactor}{4}
\providecommand{\BIBentryALTinterwordspacing}{\spaceskip=\fontdimen2\font plus
\BIBentryALTinterwordstretchfactor\fontdimen3\font minus
  \fontdimen4\font\relax}
\providecommand{\BIBforeignlanguage}[2]{{%
\expandafter\ifx\csname l@#1\endcsname\relax
\typeout{** WARNING: IEEEtran.bst: No hyphenation pattern has been}%
\typeout{** loaded for the language `#1'. Using the pattern for}%
\typeout{** the default language instead.}%
\else
\language=\csname l@#1\endcsname
\fi
#2}}
\providecommand{\BIBdecl}{\relax}
\BIBdecl

\bibitem{rassoul1}
B.~Rassouli and D.~G\"{u}nd\"{u}z, ``On perfect privacy and maximal
  correlation,'' \emph{arXiv preprint arXiv:1712.08500}, 2017.

\bibitem{makhdoumi}
A.~Makhdoumi, S.~Salamatian, N.~Fawaz, and M.~M{\'e}dard, ``From the
  information bottleneck to the privacy funnel,'' in \emph{2014 IEEE
  Information Theory Workshop (ITW 2014)}.\hskip 1em plus 0.5em minus
  0.4em\relax IEEE, 2014, pp. 501--505.

\bibitem{tishby}
N.~Tishby, F.~C. Pereira, and W.~Bialek, ``The information bottleneck method,''
  \emph{arXiv preprint physics/0004057}, 2000.

\bibitem{yamamoto}
H.~Yamamoto, ``A source coding problem for sources with additional outputs to
  keep secret from the receiver or wiretappers (corresp.),'' \emph{IEEE
  Transactions on Information Theory}, vol.~29, no.~6, pp. 918--923, 1983.

\bibitem{sankar}
L.~Sankar, S.~R. Rajagopalan, and H.~V. Poor, ``Utility-privacy tradeoffs in
  databases: An information-theoretic approach,'' \emph{IEEE Transactions on
  Information Forensics and Security}, vol.~8, no.~6, pp. 838--852, 2013.

\bibitem{dwork1}
C.~Dwork, F.~McSherry, K.~Nissim, and A.~Smith, ``Calibrating noise to
  sensitivity in private data analysis,'' in \emph{Theory of cryptography
  conference}.\hskip 1em plus 0.5em minus 0.4em\relax Springer, 2006, pp.
  265--284.

\bibitem{dwork2}
C.~Dwork, ``Differential privacy, in automata, languages and programming,''
  \emph{ser. Lecture Notes in Computer Scienc}, vol. 4052, p. 112, 2006.

\bibitem{oech}
Z.~Li, T.~J. Oechtering, and D.~G{\"u}nd{\"u}z, ``Privacy against a hypothesis
  testing adversary,'' \emph{IEEE Transactions on Information Forensics and
  Security}, vol.~14, no.~6, pp. 1567--1581, 2018.

\bibitem{borade}
S.~Borade and L.~Zheng, ``Euclidean information theory,'' in \emph{2008 IEEE
  International Zurich Seminar on Communications}.\hskip 1em plus 0.5em minus
  0.4em\relax IEEE, 2008, pp. 14--17.

\bibitem{huang}
S.-L. Huang and L.~Zheng, ``Linear information coupling problems,'' in
  \emph{2012 IEEE International Symposium on Information Theory
  Proceedings}.\hskip 1em plus 0.5em minus 0.4em\relax IEEE, 2012, pp.
  1029--1033.

\bibitem{huang2}
S.-L. Huang, C.~Suh, and L.~Zheng, ``Euclidean information theory of
  networks,'' \emph{IEEE Transactions on Information Theory}, vol.~61, no.~12,
  pp. 6795--6814, 2015.

\end{thebibliography}


\begin{thebibliography}{10}
\providecommand{\url}[1]{#1}
\csname url@samestyle\endcsname
\providecommand{\newblock}{\relax}
\providecommand{\bibinfo}[2]{#2}
\providecommand{\BIBentrySTDinterwordspacing}{\spaceskip=0pt\relax}
\providecommand{\BIBentryALTinterwordstretchfactor}{4}
\providecommand{\BIBentryALTinterwordspacing}{\spaceskip=\fontdimen2\font plus
\BIBentryALTinterwordstretchfactor\fontdimen3\font minus
  \fontdimen4\font\relax}
\providecommand{\BIBforeignlanguage}[2]{{%
\expandafter\ifx\csname l@#1\endcsname\relax
\typeout{** WARNING: IEEEtran.bst: No hyphenation pattern has been}%
\typeout{** loaded for the language `#1'. Using the pattern for}%
\typeout{** the default language instead.}%
\else
\language=\csname l@#1\endcsname
\fi
#2}}
\providecommand{\BIBdecl}{\relax}
\BIBdecl

\bibitem{yamamoto}
H.~Yamamoto, ``A source coding problem for sources with additional outputs to
  keep secret from the receiver or wiretappers (corresp.),'' \emph{IEEE
  Transactions on Information Theory}, vol.~29, no.~6, pp. 918--923, 1983.

\bibitem{sankar}
L.~Sankar, S.~R. Rajagopalan, and H.~V. Poor, ``Utility-privacy tradeoffs in
  databases: An information-theoretic approach,'' \emph{IEEE Transactions on
  Information Forensics and Security}, vol.~8, no.~6, pp. 838--852, 2013.

\bibitem{makhdoumi}
A.~Makhdoumi, S.~Salamatian, N.~Fawaz, and M.~M{\'e}dard, ``From the
  information bottleneck to the privacy funnel,'' in \emph{2014 IEEE
  Information Theory Workshop (ITW 2014)}.\hskip 1em plus 0.5em minus
  0.4em\relax IEEE, 2014, pp. 501--505.

\bibitem{dwork1}
C.~Dwork, F.~McSherry, K.~Nissim, and A.~Smith, ``Calibrating noise to
  sensitivity in private data analysis,'' in \emph{Theory of cryptography
  conference}.\hskip 1em plus 0.5em minus 0.4em\relax Springer, 2006, pp.
  265--284.

\bibitem{oech}
Z.~Li, T.~J. Oechtering, and D.~G{\"u}nd{\"u}z, ``Privacy against a hypothesis
  testing adversary,'' \emph{IEEE Transactions on Information Forensics and
  Security}, vol.~14, no.~6, pp. 1567--1581, 2018.

\bibitem{issa}
I.~{Issa}, S.~{Kamath}, and A.~B. {Wagner}, ``An operational measure of
  information leakage,'' in \emph{2016 Annual Conference on Information Science
  and Systems (CISS)}, March 2016, pp. 234--239.

\bibitem{Calmon2}
H.~{Wang}, L.~{Vo}, F.~P. {Calmon}, M.~{M\'{e}dard}, K.~R. {Duffy}, and
  M.~{Varia}, ``Privacy with estimation guarantees,'' \emph{IEEE Transactions
  on Information Theory}, vol.~65, no.~12, pp. 8025--8042, Dec 2019.

\bibitem{asoodeh1}
S.~Asoodeh, M.~Diaz, F.~Alajaji, and T.~Linder, ``Information extraction under
  privacy constraints,'' \emph{Information}, vol.~7, no.~1, p.~15, 2016.

\bibitem{houi}
L.~{Zhou}, M.~T. {Vu}, T.~J. {Oechtering}, and M.~{Skoglund}, ``Fundamental
  limits for biometric identification systems without privacy leakage,'' in
  \emph{2019 57th Annual Allerton Conference on Communication, Control, and
  Computing (Allerton)}, Sep. 2019, pp. 1105--1112.

\bibitem{deniz6}
B.~{Rassouli} and D.~{G\"{u}nd\"{u}z}, ``On perfect privacy,'' in \emph{2018
  IEEE International Symposium on Information Theory (ISIT)}, June 2018, pp.
  2551--2555.

\bibitem{gun}
S.~{Sreekumar} and D.~{G\"{u}nd\"{u}z}, ``Optimal privacy-utility trade-off
  under a rate constraint,'' in \emph{2019 IEEE International Symposium on
  Information Theory (ISIT)}, July 2019, pp. 2159--2163.

\bibitem{sankar2}
J.~{Liao}, O.~{Kosut}, L.~{Sankar}, and F.~P. {Calmon}, ``Tunable measures for
  information leakage and applications to privacy-utility tradeoffs,''
  \emph{IEEE Transactions on Information Theory}, vol.~65, no.~12, pp.
  8043--8066, Dec 2019.

\bibitem{deniz4}
B.~{Rassouli}, F.~{Rosas}, and D.~{G\"{u}nd\"{u}z}, ``Latent feature disclosure
  under perfect sample privacy,'' in \emph{2018 IEEE International Workshop on
  Information Forensics and Security (WIFS)}, Dec 2018, pp. 1--7.

\bibitem{asoodeh3}
S.~{Asoodeh}, M.~{Diaz}, F.~{Alajaji}, and T.~{Linder}, ``Estimation efficiency
  under privacy constraints,'' \emph{IEEE Transactions on Information Theory},
  vol.~65, no.~3, pp. 1512--1534, March 2019.

\bibitem{Calmon1}
F.~P. {Calmon} and N.~{Fawaz}, ``Privacy against statistical inference,'' in
  \emph{2012 50th Annual Allerton Conference on Communication, Control, and
  Computing (Allerton)}, Oct 2012, pp. 1401--1408.

\bibitem{7888175}
Y.~O. {Basciftci}, Y.~{Wang}, and P.~{Ishwar}, ``On privacy-utility tradeoffs
  for constrained data release mechanisms,'' in \emph{2016 Information Theory
  and Applications Workshop (ITA)}, Jan 2016, pp. 1--6.

\bibitem{nekouei2}
E.~Nekouei, T.~Tanaka, M.~Skoglund, and K.~H. Johansson,
  ``Information-theoretic approaches to privacy in estimation and control,''
  \emph{Annual Reviews in Control}, 2019.

\bibitem{Johnson}
M.~P. {Johnson}, L.~{Zhao}, and S.~{Chakraborty}, ``Achieving pareto-optimal
  mi-based privacy-utility tradeoffs under full data,'' \emph{IEEE Journal of
  Selected Topics in Signal Processing}, vol.~12, no.~5, pp. 1093--1105, Oct
  2018.

\bibitem{Total}
B.~{Rassouli} and D.~{G\"{u}nd\"{u}z}, ``Optimal utility-privacy trade-off with
  total variation distance as a privacy measure,'' \emph{IEEE Transactions on
  Information Forensics and Security}, vol.~15, pp. 594--603, 2020.

\bibitem{jende}
B.~Razeghi, F.~Calmon, D.~G\"{u}nd\"{u}z, and S.~Voloshynovskiy, ``On perfect
  obfuscation: Local information geometry analysis,'' \emph{arXiv preprint
  arXiv:2009.04157}, 9 Sep 2020.

\bibitem{borade}
S.~Borade and L.~Zheng, ``Euclidean information theory,'' in \emph{2008 IEEE
  International Zurich Seminar on Communications}.\hskip 1em plus 0.5em minus
  0.4em\relax IEEE, 2008, pp. 14--17.

\bibitem{huang}
S.-L. Huang and L.~Zheng, ``Linear information coupling problems,'' in
  \emph{2012 IEEE International Symposium on Information Theory
  Proceedings}.\hskip 1em plus 0.5em minus 0.4em\relax IEEE, 2012, pp.
  1029--1033.

\bibitem{berger}
T.~{Berger} and R.~W. {Yeung}, ``Multiterminal source encoding with encoder
  breakdown,'' \emph{IEEE Transactions on Information Theory}, vol.~35, no.~2,
  pp. 237--244, March 1989.

\bibitem{wu2017lecture}
Y.~Wu, ``Lecture notes for ece598yw: Information-theoretic methods for
  high-dimensional statistics,'' 2017.

\bibitem{el2011network}
A.~El~Gamal and Y.-H. Kim, \emph{Network information theory}.\hskip 1em plus
  0.5em minus 0.4em\relax Cambridge university press, 2011.

\end{thebibliography}
\begin{IEEEbiography}[{\includegraphics[width=1in,height=1.25in,clip,keepaspectratio]{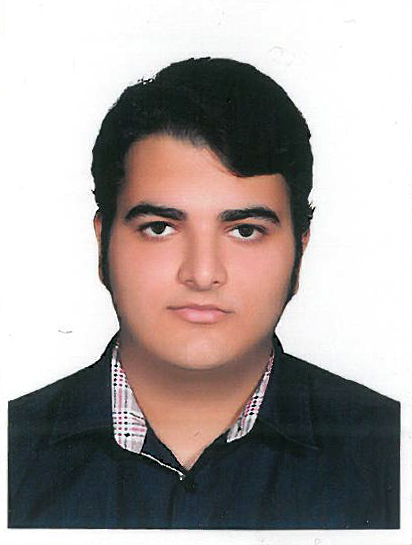}}]{Amirreza Zamani}
	received the B. Tech. degree in electrical engineering from the University of Tehran, Iran, in 2016, the M.Sc. degree from Sharif university of Technology, Iran, in 2018. He is presently a Ph.D. student at KTH Royal Institute of Technology, Stockholm, Sweden. His research interests include statistical inference, information theory, information privacy and security. 
\end{IEEEbiography}
\begin{IEEEbiography}[{\includegraphics[width=1in,height=1.25in,clip,keepaspectratio]{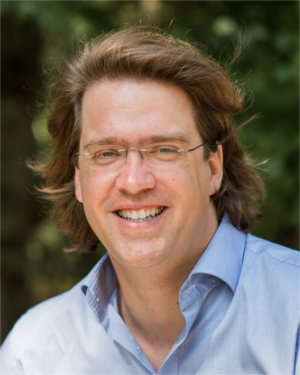}}]{Tobias J. Oechtering}
(S’01-M’08-SM’12) received his Dipl-Ing degree in Electrical Engineering and Information Technology in 2002 from RWTH Aachen University, Germany, his Dr-Ing degree in Electrical Engineering in 2007 from the Technische Universit\"at Berlin, Germany. In 2008 he joined KTH Royal Institute of Technology, Stockholm, Sweden and has been a Professor since 2018. In 2009, he received the ``F\"orderpreis 2009” from the Vodafone Foundation.

Dr. Oechtering is currently Senior Editor of IEEE Transactions on Information Forensic and Security since May 2020 and served previously as Associate Editor for the same journal since June 2016, and IEEE Communications Letters during 2012-2015. He has served on numerous technical program committees for IEEE sponsored conferences, and he was general co-chair for IEEE ITW 2019. His research interests include physical layer privacy and security, statistical signal processing, communication and information theory, as well as communication for networked control. 

\end{IEEEbiography}
\begin{IEEEbiography}[{\includegraphics[width=1in,height=1.25in,clip,keepaspectratio]{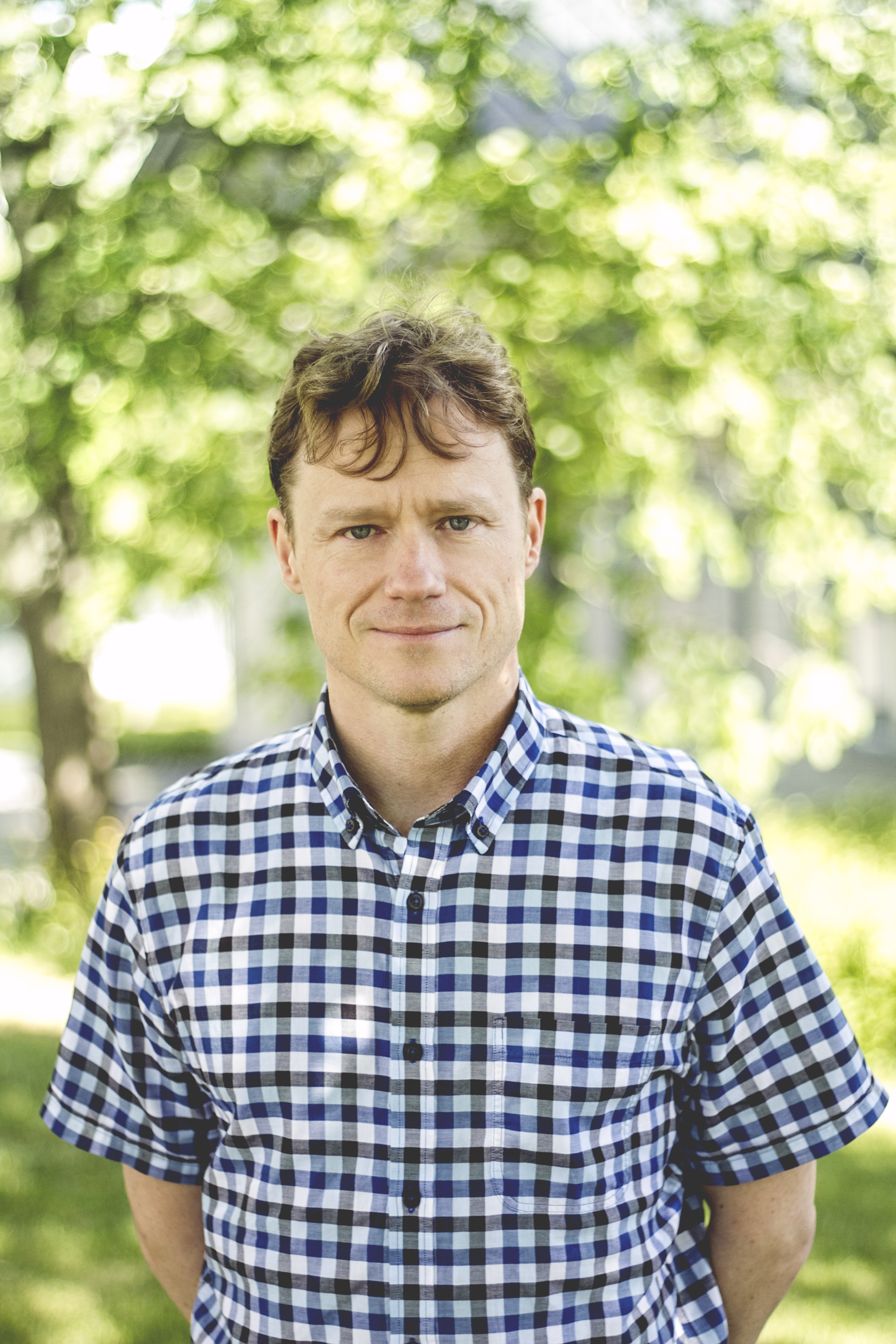}}]{Mikael Skoglund}
	(S'93-M'97-SM'04-F'19) received the Ph.D.~degree in
	1997 from Chalmers University of Technology, Sweden.  In 1997, he joined the Royal Institute of Technology (KTH), Stockholm, Sweden, where he was appointed to the Chair in Communication Theory in 2003.  At KTH, he heads the Division of Information Science and Engineering, and the Department of Intelligent Systems.
	
	Dr.~Skoglund has worked on problems in source-channel coding, coding and transmission for wireless communications, Shannon theory, information and control, and statistical signal processing. He has authored and co-authored more than 160 journal and 380 conference papers.
	
	Dr.~Skoglund is a Fellow of the IEEE. During 2003--08 he was an associate editor for the IEEE Transactions on Communications and during 2008--12 he was on the editorial board for the IEEE Transactions on Information Theory. He has served on numerous technical program committees for IEEE sponsored conferences, he was general co-chair for IEEE ITW 2019, and he will serve as TPC co-chair for IEEE ISIT 2022.
\end{IEEEbiography}
\end{document}